\documentclass[journal,comsoc]{IEEEtran} 
\usepackage[T1]{fontenc}
\usepackage{graphicx}
\usepackage{subfigure}
\usepackage{cite}
\usepackage{balance}
\usepackage{amsthm}
\usepackage{epsfig}
\usepackage{latexsym}
\usepackage{amsfonts}
\usepackage{here}
\usepackage{rawfonts}
\usepackage{calc}
\usepackage{url}
\usepackage{enumerate}
\usepackage{color}
\usepackage[tbtags]{amsmath}
\usepackage{amssymb}
\usepackage{upref}
\usepackage{epic,eepic}
\usepackage{epstopdf}
\usepackage{times}
\usepackage{dsfont}
\usepackage{comment}
\usepackage{cite}
\allowdisplaybreaks

\newtheorem{theorem}{\bf Theorem}

\newtheorem{lemma}{\bf Lemma}

\newtheorem{definition}{\bf Definition}

\newcounter{step}
\newlength{\totlinewidth}
  {\end{list}%
  \rule{\linewidth}{1pt}}
\newcounter{substep}

  {\end{list}}

\newlength{\aligntop}
\newlength{\alignbot}
 \makeatletter
\renewenvironment{align}{%
  \vspace{\aligntop}
  \start@align\@ne\st@rredfalse\m@ne
}{%
  \math@cr \black@\totwidth@
  \egroup
  \ifingather@
    \restorealignstate@
    \egroup
    \nonumber
    \ifnum0=`{\fi\iffalse}\fi
  \else
    $$%
  \fi
  \ignorespacesafterend%
  \vspace{\alignbot}\par\noindent
} \makeatother

\IEEEoverridecommandlockouts

\begin{document}
\clearpage
\title{Stochastic Coalitional Games for Cooperative Random Access in M2M Communications}
%

\author{Mehdi Naderi Soorki, Walid Saad, Mohammad Hossein Manshaei, and Hossein Saidi
\thanks{M. Naderi Soorki, M. H. Manshaei, and H. Saidi are with the Department of Electrical and Computer Engineering, Isfahan University of Technology, Isfahan 84156-83111, Iran. e-mail: m.naderisoorki@ec.iut.ac.ir, \{manshaei,hsaidi\}@cc.iut.ac.ir.}
\thanks{W. Saad is with the Wireless@VT, Bradley Department of Electrical and Computer Engineering, Virginia Tech, Blacksburg, VA 24061, USA, e-mail: walids@vt.edu}
\thanks{This research was supported by the Office of Naval Research (ONR) under Grant N00014-15-1-2709.}
\thanks{Manuscript received 13 April 2016; revised 25 Sep. 2016, second revised 15 Jan. 2017, third revised 24 April 2017, accepted 15 June 2017.}
}

\maketitle
\thispagestyle{empty}
\vspace{0cm}
\begin{abstract}
In this paper, the problem of random access contention between machine type devices (MTDs) in the uplink of a wireless cellular network is studied. In particular, the possibility of forming cooperative groups to coordinate the MTDs' requests for the random access channel (RACH) is analyzed. The problem is formulated as a stochastic coalition formation game in which the MTDs are the players that seek to form cooperative coalitions to optimize a utility function that captures each MTD's energy consumption and time-varying queue length. Within each coalition, an MTD acts as a coalition head that sends the access requests of the coalition members over the RACH. One key feature of this game is its ability to cope with stochastic environments in which the arrival requests of MTDs and the packet success rate over RACH are dynamically time-varying. The proposed stochastic coalitional is composed of multiple stages, each of which corresponds to a coalitional game in stochastic characteristic form that is played by the MTDs at each time step. To solve this game, a novel distributed coalition formation algorithm is proposed and shown to converge to a stable MTD partition. Simulation results show that, on the average, the proposed stochastic coalition formation algorithm can reduce the average fail ratio and energy consumption of up to 36\% and 31\% for a cluster-based distribution of MTDs, respectively, compared to a noncooperative case. Moreover, when the MTDs are more sensitive to the energy consumption (queue length), the coalitions' size will increase (decrease).
\end{abstract}

{\small \emph{Index Terms}--- Game Theory; Machine-to-Machine Communications; Internet of Things; Coalitional Games.}

\section{Introduction}
\label{sec:Intro}
Machine-to-machine (M2M) communication between machine type devices (MTDs) such as sensors or wearables will lie at the heart of tomorrow's Internet of Things (IoT) system~\cite{ghavimi2014m2m}. In order to support massive M2M communications, there is a need for a reliable wireless infrastructure. In this respect, cellular networks provide an ideal platform for M2M communications, due to their proven effectiveness and reliability. However, deploying M2M over cellular networks such as LTE faces many challenges that range from network deployment to resource allocation and multiple access~\cite{al2015internet,ghavimi2014m2m,cheng2015d2d,laya2014random,shafiq2012first,hussain2014multi}.

In particular, in cellular LTE systems, whenever a device intends to access the network, it begins by following a random access (RA) procedure that is done \emph{before the resource allocation phase} as explained in~\cite{laya2014random} and~\cite{Evolved2011Universal}. During the RA procedure, each device needs to send a preamble signal over a special channel, known as the physical random access channel (RACH), which is used to transmit an initial preamble. In the RA procedure of a cellular network, the RACH is formed by a periodic sequence of allocated time-frequency resources, called random access (RA) slots. These slots are reserved in the uplink channel of the network for the transmission of access requests~\cite{Evolved2011Universal}. Thus, the transmission of preambles that shows the requests made by MTDs for uplink resource, is synchronized during the RA slots. Each device selects the preamble uniformly from the available preambles~\cite{laya2014random}. LTE typically uses a contention-based random access procedure for the initial association to the network, for the request of resources for transmission, and for re-establishing a connection upon failure~\cite{laya2014random} and~\cite{Evolved2011Universal}. The contention-based RA procedure consists of a four-message handshake between any device such as MTDs and the base station (BS) in order to successfully transmit an initial preamble. The RA procedure in existing cellular systems has been mainly designed for human-to-human (H2H) communication scenarios in which the amount of uplink (UL) traffic is normally lower than the downlink (DL) traffic. In contrast, M2M applications will produce significantly more UL traffic than in the downlink~\cite{ghavimi2014m2m}. In an M2M scenario, during the RA process of LTE, a large number of MTDs will simultaneously attempt to access a
shared preamble. This can have several drawbacks such as to a low random access success rate, the waste of radio resources, packet loss, latency, and extra energy consumption as pointed out in~\cite{ghavimi2014m2m} and~\cite{laya2014random}.

Several recent works have proposed new techniques for reducing RA congestion in M2M scenarios such as in~\cite{hasan2013random,ilori2015random,laya2014random,thomsen2013code,yang2012performance,kimenhanced, jang2014spatial}. A number of such works, such as~\cite{laya2014random,ilori2015random}, and~\cite{thomsen2013code}, focus on the RA process that involves preamble selection, designing a new preamble sequence, and efficient preamble allocation. In~\cite{yang2012performance}, the authors propose a new backoff algorithm while the work in~\cite{cheng2011prioritized} introduces a prioritized RA architecture. The objective in these works is to improve RA efficiency. On the other hand, there has been a number of recent works such as~\cite{hussain2014multi,kimenhanced,jang2014spatial,fu2014group,lien2011toward,lee2013feasibility,tu2011energy,ho2012energy}, and~\cite{wei2012joint} that focus on how to cluster MTDs in an efficient way so as to decrease the load over the RACH. An enhanced RA scheme based on spatial grouping for reusable preamble allocation is proposed in~\cite{kimenhanced} and~\cite{jang2014spatial}. This scheme reuses the preamble resources based on spatial grouping during the RA procedure. A group mobility management mechanism is studied in~\cite{fu2014group} using which MTDs are grouped based on the similarity of their mobility patterns at the location database, and only the leader machine performs mobility management. Clustering techniques based on quality-of-service (QoS) requirements have been proposed~\cite{lien2011toward,tu2011energy} and~\cite{ho2012energy}, for power allocation and energy-efficient M2M communications. Optimal cluster formation and power control for maximizing throughput and minimizing transmit power are derived using mixed-integer non-linear programming in~\cite{hussain2014multi} and~\cite{wei2012joint}.

{These existing works for M2M clustering such as in~\cite{hussain2014multi,kimenhanced,jang2014spatial,fu2014group,lien2011toward,lee2013feasibility,tu2011energy,ho2012energy}, and~\cite{wei2012joint} have mainly focused on clustering using spatial metrics such as the distance between MTDs and some basic QoS metrics. Moreover, these works rely on centralized algorithms that are normally used to find the optimal clustering of the network. However, a centralized approach requires collecting a significant amount of information on the random arrival of requests at the level of the MTDs as well as on the random collisions that can occur over the RACH. In an M2M network, this information collection must be updated every time slot due to the stochastic changes in the arrival of requests to the MTDs and the possible collisions. If done in a centralized manner, such a dynamic information update will significantly increase the signaling overhead over the uplink of the M2M network and, thus, will not be practical. In addition to signaling overhead, a centralized approach to coalition formation is generally known to be NP-complete as shown in~\cite{sandholm1999coalition}, especially for large number of MTDs. The complexity of such centralized approach grows exponentially with the number of MTDs because the number of all possible partitions for MTD set given by a value known as the Bell number~\cite{sandholm1999coalition}. Thus, to decrease the complexity and signaling overhead it is highly desirable to equip the MTDs with distributed cooperative strategies that require little or no reliance on centralized entities such as base stations. Moreover, clustering MTDs in a practical cellular network must account not only for spatial metrics and QoS such as in~\cite{hussain2014multi,kimenhanced,jang2014spatial,fu2014group,lien2011toward,lee2013feasibility,tu2011energy,ho2012energy} and~\cite{wei2012joint}, but also for the stochastic changes in the M2M communication environment. Thus, the need for a stochastic coalition formation approach results from the fact that, in practice, an M2M communication network is highly dynamic and stochastic in nature. This stochastic nature stems from various features of the M2M environment such as random arrival of access requests to the MTDs and random preamble collision over the RACH. {Note that modeling and capturing the various dynamics of the M2M system, such as the random requests that arrive at the MTDs or random collision over RACH is very challenging even for a single-cell scenario. In fact, this dynamic clustering problem has not been considered in any of the existing literature on M2M such as the works in~\cite{hussain2014multi,kimenhanced,jang2014spatial,fu2014group,lien2011toward,lee2013feasibility,tu2011energy,ho2012energy} and~\cite{wei2012joint} that also consider a single BS.} The advantages of using a stochastic coalition formation approach are: 1) distributed solutions do not require any database that records information such as the MTDs' locations or stochastic environment changes such as the random arrival of requests at the level of the MTDs or the collisions that can occur over the RACH, 2) the signaling overhead for updating dynamically varying information decreases in a coalitional game solution due to the fact that the MTDs will autonomously perform coalition formation to adapt to the new changes without any need to send any information to a centralized controller, 3) the complexity of clustering MTDs is more manageable in a distributed coalitional game solution because the MTDs will individually perform distributed coalition formation and, unlike in the centralized approach, there is no need to search over all the partitions of the MTDs' set, and 4) a coalitional game formulation allows understanding how each MTD can make its own decision on forming cooperative group in a self-configuring M2M network.} Here, we note that the use of overlapping coalition formation approaches such as the ones in~\cite{lu2014layered} and~\cite{7066970} is not suitable for M2M communication scenarios. Overlapping coalition game models are useful in problems in which devices can further improve the system performance and efficiency by splitting their coalition membership between multiple, overlapping coalitions~\cite{lu2014layered} and~\cite{7066970}. In the cooperative M2M random access problem, having an overlap between two coalitions of MTDs will not lead to sharing additional preambles between the overlapping coalitions. Moreover, the incoming packet rate of the queues of devices belonging to overlapping coalitions will increase. Consequently, as proposed in this work, one must adopt the more tractable non-overlapping coalition formation approaches for M2M clustering.

The main contribution of this paper is to analyze the RA procedure for M2M communications and design a new coalition formation protocol using which the MTDs can autonomously form clusters or coalitions in the presence of stochastic arrival requests and a stochastic number of successfully transmitted packets over the RACH. We formulate the problem as a stochastic coalition formation game in which the MTDs are the players. In this game, the MTDs seek to cooperate with one another in order to coordinate their RA and use of the RACH. In particular, within each coalition, a coalition head sends the access requests of the coalition's members. The performance of each coalition is captured via a utility function that reflects the number of requests that the coalition members want to send over the RACH and the energy consumption of its members during each time slot. To solve this game, we propose an algorithm that enables the MTDs to form the optimal coalitions while optimizing a utility function that captures stochastic changes such as the arrival requests of MTDs and the packet success rate of the RACH. Under these stochastic changes, we show that the proposed algorithm can reach a stable partition, if the MTDs are sufficiently farsighted and they value future payoffs more than the present ones. For this game, we compute the required threshold of the farsighted level of the MTDs that is needed to form stable coalitions under proposed algorithm. Simulation results show that the proposed approach can reduce the fail ratio and energy consumption compared to a traditional noncooperative random access model. The results show that, on the average, the proposed stochastic cooperative random access model provides a reduction of the fail ratio and energy power up to 36\% and 31\% for a cluster-based distribution of MTDs, respectively, compared to a noncooperative case. We note that, although coalitional game theory has been used in many works related to wireless communication such as~\cite{han2012game,5230848} and~\cite{7395023}, to the best of our knowledge, none of these existing works has developed a stochastic game model, in general, and for M2M communication, in particular. In summary, the novelty of our contribution, compared to existing work, lies in the following key points:
\begin{itemize}
\item We develop a new, M2M-specific model for the stochastic value function of an M2M coalition. This model shows that, after forming coalitions, the value function of the coalitions will change during the next time slots and the players will become uncertain about their payoff in future time slots. Then, in a given time slot, we use a new coalitional game class, known as games in \emph{stochastic characteristic function form}.
\item In addition to modeling the cooperation of MTDs during one time slot as a coalitional game in stochastic characteristic function form, we modeled the cooperation of MTDs during different time slots as a \emph{stochastic coalition game}. In each stage of the stochastic coalition game, the cooperation of MTDs is modeled using the game mentioned in the previous bullet.
\item For finding stable coalitions under unknown stochastic changes in the value function, we have developed a novel coalition formation algorithm that explicitly accounts for the presence of a discount factor $\delta$ in the value function.
\end{itemize}

The rest of this paper is organized as follows. Section~\ref{Sec:Non-Model} presents the noncooperative random access model in a cellular LTE network. In Section~\ref{Sec:Gam-Model}, we model the problem using stochastic coalition formation in games. Simulation results are presented and analyzed in Section~\ref{Sec:Performance}. Finally, conclusions are drawn in Section~\ref{Sec:Conclusion}.
\section{System Model}
\label{Sec:Non-Model}
Consider a cellular network composed of one BS and a set $\mathcal{M}$ of $M$ MTDs that seek to access the network's uplink resources to send their data. In this network, the RACH includes $\mu$ preambles. In each time slot, each MTD will randomly transmit one of these preambles to the BS. We consider an infinite number of discrete time slots $\{1,2,...,t,...\}$ with duration $T$ for each RA slot. In each RA slot, every MTD $m \in \mathcal{M}$ transmits its access request with a fixed probability $p$. This is a practical assumption when considering the access barrier algorithm that is used to reshape and distribute the traffic over the RACH, as discussed in~\cite{cheng2011prioritized} and~\cite{7029666}. Since our model focuses on the RA process, we do not consider human type devices as their impact will be equivalent to the MTDs.

When the MTDs are acting in a noncooperative manner, if an MTD $m$ selects an RA preamble while other MTDs that want to transmit RA requests do not select that specific preamble, then MTD $m$ can successfully transmit an RA request since there will be no collisions. Hence, if only MTD $m$ submits an RA request then the probability of its successful transmission will be equal to the probability of data transmission $p$. Thus, the probability that an MTD $m$ successfully transmits an RA request in a noncooperative manner is given by:
\begin{equation}\label{prob}
P_s^\mathcal{M}=p(1-p)^{M-1}+\sum_{j=2}^{M}\frac{1}{\mu}\left(1-\frac{1}{\mu} \right)^{j-1} P_M(j),
\end{equation}
where $P_M(j)=\frac{M!}{j!(M-j)!}p^{j}(1-p)^{M-j}$ is the probability that $j$ MTDs out of a total of $M$ MTDs transmit their access requests during the current RA time slot. In (\ref{prob}), $\frac{1}{\mu}\left(1-\frac{1}{\mu}\right)^{j-1}$ is the probability that MTD $m$ selects a preamble which is different from the preambles that were selected by the other $j-1$ MTDs.

Due to the random nature of the arrival requests at each MTD and the departure requests over the RACH, we define a discrete-time queueing system with $Q_{m,t}$ being the number of requests that MTD $m$ has buffered at time slot $t$. We assume that the maximum number of requests in this queue is $K$. When MTDs do not cooperate, the change in the queue length of each MTD $m$ at each time slot is given by:
\begin{equation}
Q_{m,t+1}=\max\{Q_{m,t}-d_{m,t}+a_{m,t},K\},
\label{queue_mac}
\end{equation}
where $a_{m,t}$ and $d_{m,t}$ represent, respectively, the number of arrival and departure requests for MTD $m$  conditioned on the packet success rate over the RACH at time slot $t$. $Q_{m,t}$ can be modeled as a $G/G/1$ queue which represents the queue length in a system with a single server where inter-arrival times have a general (arbitrary) distribution and service times have a (different) general distribution~\cite{cooper1981introduction}. In particular, we have $\textrm{Pr}(a_{m,t}=1)=p$ and $\textrm{Pr}(d_{m,t}=1)=P_s^\mathcal{M}$. {(\ref{queue_mac}) models the $B$-bit preambles as a queue at the MAC layer. Some packets can be dropped due to the limited length of the buffer at the MAC layer or due to collision over the RACH. At the MAC layer, the loss is defined by the packet loss rate. For sending one $B$-bit preamble from the queue of the MAC layer, each MTD must transmit $B$ bits over the RA slot.}

For simplicity, we assume that each MTD transmits a single, fixed-sized packet request of size $B$ bits to the BS with transmit power $P_{LR}^z$. Let $Z_{LR}$ be the number of subcarriers in the network, with each subcarrier having a bandwidth $B_z$. Let $h_{m}^z=H_0|d_m|^{-\nu}\xi$ be the channel gain for cellular link between MTD $m$ and BS, where $H_0$ is the path loss constant, $d_m$ is the distance between MTD $m$ and the BS, $\nu$ is the path loss exponent, and $\xi$ is the flat fading Rayleigh parameter with mean $1$. The achievable rate of the cellular link between MTD $m$ and BS for subcarrier $z$ can be given by:
\begin{equation}
R_{m}=B_z\log_2(1+\frac{P_{LR}^zh_{m}^z}{N_0+\sum\limits_{n\neq m,z'=z}{h_{n}^{z'}P_{LR}^{z'}}}),
\label{bitrate_physical}
\end{equation}
where $\sum\limits_{n\neq m,z'=z}{h_{n}^{z'}P_{LR}^{z'}}$ is the interference received from other MTDs over the cellular link. {(\ref{bitrate_physical}) shows the achievable rate over the resource of RA slot in the physical layer of the cellular link. The loss in the physical layer is due to channel gain, noise, and interference which are captured by (\ref{bitrate_physical}).} The time needed to send the $B$-bit packet will be $\frac{B}{R_{m}}$. We introduce a power allocation mechanism that allocates power over the cellular links to guarantee $\frac{B}{R_m}\leq T$ at each time slot. Thus, $E^z_{LR}=P^z_{LR}\frac{B}{R_{m}}$. The energy consumption per-packet over the cellular link ${E}_{LR}$ is defined as the total energy spent during the RA procedure until the successful transmission of the first packet over the RACH~\cite{laya2014random}. In such a noncooperative manner, the average per-request energy consumption of each MTD $m$ is given by the series: $\bar{E}_{LR}^{M}=P_s^\mathcal{M} E^z_{LR}+P_s^\mathcal{M} (1-P_s^\mathcal{M}) 2 E^z_{LR}+P_s^\mathcal{M} (1-P_s^\mathcal{M})^{2} 3 E^z_{LR}+...$ which can be written as follow:
\begin{align}
&\bar{E}_{LR}^{M}=\sum_{t=1}^{\infty} P_s^\mathcal{M} (1-P_s^\mathcal{M})^{t-1} t E^z_{LR}=\nonumber\\
& \sum_{t=1}^{\infty} P_s^\mathcal{M} (1-P_s^\mathcal{M})^{t-1} E^z_{LR}+
\sum_{t=2}^{\infty} P_s^\mathcal{M} (1-P_s^\mathcal{M})^{t-1} (t-1) E^z_{LR}=\nonumber\\
& \sum_{t=1}^{\infty} P_s^\mathcal{M} (1-P_s^\mathcal{M})^{t-1} E^z_{LR}+\nonumber\\
&(1-P_s^\mathcal{M}) \big(\sum_{t=2}^{\infty} P_s^\mathcal{M} (1-P_s^\mathcal{M})^{t-2} (t-1) E^z_{LR} \big)= \frac{E^z_{LR}}{P_s^\mathcal{M}}.
\label{average energy}
\end{align}
%

Here, if the probability of successful transmission over RACH is equal to 1, $P_s^\mathcal{M}=1$, the average per-request energy consumption of each MTD $m$ will be $E^z_{LR}$ because there is no collision over the RACH.

The massive amount of incoming access requests stemming from MTDs can lead to a low packet success rate of the RACH, increased packet loss, intolerable latency, and increased energy consumption~\cite{laya2014random,larmo2012ran}. For M2M communications, the number of access requests that the MTDs can send is a more important metric than the bit rate or throughput. This is due to the fact that, the MTDs usually need to send data at a very low bit rate (M2M traffic payload size is small) because the size of the messages is generally very short in M2M applications (e.g. very few bits coming from a smart meter or sensor, or even just 1 bit used to inform of the existence or absence of a given event)~\cite{laya2014random}.  Thus, in the presence of a massive number of MTDs, one must develop new approaches to decrease packet loss and energy consumption.  {Consequently, the goal for each MTD $m$ is to minimize two objectives: queue length $Q_m(t)$ and energy consumption $\bar{E}_{TL}$. For each MTD $m$, the objective can be viewed as a multi-objective optimization problem in which the MTD must balance the tradeoff between queue length and energy consumption. According to linear scalarization technique~\cite{hwang2012multiple}, the single objective of a multi-objective optimization scalarized problem is the weighted summation of multiple objectives, with the weights being the parameters of the scalarization. Thus, for each MTD $m$ the objective can be given by:}
\begin{equation}
\min  (\alpha_m Q_{m,t}+\beta_m \bar{E}_{TL}^{M}),
\label{singel tone}
\end{equation}
{where $\alpha_m$ and $\beta_m$ are the weights or preferences of MTD $m$ with respect to the queue length and average energy consumption per-packet, respectively. These parameters are used to adjust the \emph{scales and units} of the queue length (number of packets) and the energy consumption (joule).}
%

Formally, a coalition $\mathcal{S} \subseteq \mathcal{M}$ is defined as a subset of $\mathcal{M}$, while a \emph{partition} $\Pi_t=\{\mathcal{S}_1,\mathcal{S}_2,...,\mathcal{S}_{|\Pi_t|}\}$ is a set of mutually disjoint coalitions that span all of $\mathcal{M}$ at time slot $t$.  Whenever a coalition $\mathcal{S}$ of MTDs forms, its members exchange the access requests in their queues , $Q_{m,t}$, over short-range (SR) M2M channels. Thus, the queue of requests in the coalition $\mathcal{S}_i$ is equal to $Q_{\mathcal{S}_i,t}=\sum_{m\in \mathcal{S}_i} Q_{m,t}$. Let $\mathcal{Q}(t)=\{Q_1,Q_2,...,Q_{|\Pi_t|}\}$ be a set that represents the lengths of the queues of all coalitions at time slot $t$. Once an MTD in $\mathcal{S}$ is selected as a coalition head, it will be responsible to send the access request of its coalition's members over the RACH. This coalition head will be referred to as a machine type head (MTH). In our model, the set of MTHs of all coalitions within a partition $\Pi_t$ is $\mathcal{H}_t=\{H_1,H_2,...,H_{|\Pi_t|}\}$ at time slot $t$.

Let $a_{\mathcal{S}_i,t}$ and $d_{\mathcal{S}_i,t}$ be, respectively, the arrival rate of requests to $\mathcal{S}_i$ and the departure rate of requests from coalition $\mathcal{S}_i$ at time slot $t$. During each time slot, $a_{\mathcal{S}_i,t}$ can change from $0$, which means that none of the MTDs in the coalition $\mathcal{S}_i$ is sending an access request, to $|\mathcal{S}_i|$ which implies that all of the MTDs in the coalition $\mathcal{S}_i$ need to send an access request. The probability that $a_{\mathcal{S}_i,t}=n$, where $0\leq n\leq|\mathcal{S}_i|$, is given by:
\begin{equation}
\textrm{Pr}(a_{\mathcal{S}_i,t}=n)=\frac{|\mathcal{S}_i|!}{n!(|\mathcal{S}_i|-n)!}p^n(1-p)^{|\mathcal{S}_i|-n}.
\label{arrival_prob}
\end{equation}

Here, $d_{\mathcal{S}_i,t}$ can be $0$ or $1$, because the head of coalition $\mathcal{S}_i$ can successfully send data or a RACH collision may occur, during each time slot. The probability that $d_{\mathcal{S}_i,t}=1$ is:
\begin{equation}
\textrm{Pr}(d_{\mathcal{S}_i,t}=1)=P_s^{\mathcal{H}_t}=\frac{1}{\mu}\left(1-\frac{1}{\mu}\right)^{|\mathcal{H}_t|-1},
\label{exact success rate}
\end{equation}
where $|\mathcal{H}_t|-1$ is the number of all coalition heads except that of coalition $\mathcal{S}_i$. The queues of all of these $|\mathcal{H}_t|-1$ heads are not empty and all of them want to access to RACH during time slot. Then, the evolution of queue length of coalition $\mathcal{S}_i$ can be given by:
\begin{equation}
Q_{\mathcal{S}_i,t+1}=\max\{Q_{\mathcal{S}_i,t}-d_{\mathcal{S}_i,t}+a_{\mathcal{S}_i,t},k\}.
\label{coalition queue}
\end{equation}

The evolution of the queue length at each MTD $m$ in the coalition $\mathcal{S}_i$, can be given by:
\begin{equation}
Q_{m,t+1}=\max\{Q_{m,t}-w_m^q d_{\mathcal{S}_i,t}+a_{m,t},k\},
\label{member coalition queue}
\end{equation}
where $w_m^q$ is a coefficient that is related to fair scheduling. In this regard, each coalition head applies a fair scheduling scheme to select the request of its coalition's members and, subsequently, send it over RACH. For example, if we use a simple round-robin scheme in which the head will collect sequentially one request from each member in the coalition $\mathcal{S}_i$ from queue $Q_i$, the coefficient in (\ref{member coalition queue}) will be $w_m^q=\frac{1}{|\mathcal{S}_i|}$.

We define $P_{SR}^z$ as transmission power over a direct M2M link between a MTD and MTH. Let $Z_{SR}$ be the subcarriers for M2M links, with each subcarrier having a bandwidth $B_z$. Further, we let $h_{im}^z=H_0|d_{im}|^{-\nu}\xi$ as the channel gain for the M2M link between MTD $m$ and head of the coalition $i$ where $d_{im}$ is the distance between MTDs $i$ and $m$. The achievable rate of the M2M link between MTD $m$ and head of the coalition $i$ for subcarrier $z$ can be given by:
\begin{equation}
R_{im}=B_z\log_2\left(1+\frac{P_{SR}^zh_{im}^z}{N_0+\sum\limits_{n\neq m,z'=z}{h_{nm}^{z'}P_{SR}^{z'}}}\right).
\end{equation}

Here, we assume that M2M communications occur over subcarriers that are orthogonal to the uplink cellular communication links. Thus, there is no interference between uplink cellular links and M2M links. Since M2M links are shared among MTDs to transmit data to the MTHs, there is interference among M2M links. The decoding success probability of MTHs depends on the interference received from other MTDs. Assuming an interference limited regime, each MTH listens successfully receives the MTD packets during one time slot if the channel gains between members of each coalition and allocated power should be high enough to guarantee condition $\frac {B}{R_{im}}\leq T$. The per-packet energy consumption over M2M link ${E}_{SR}$ is defined as the total energy consumed by an MTD to transmit a single fixed-sized packet request to the MTH of its coalition. It is given by ${E}_{SR}=P_{SR}^z\frac{B}{R_{im}}$.

\begin{figure}[t]
\centering
\includegraphics[width=8.0cm,height=2.1in]{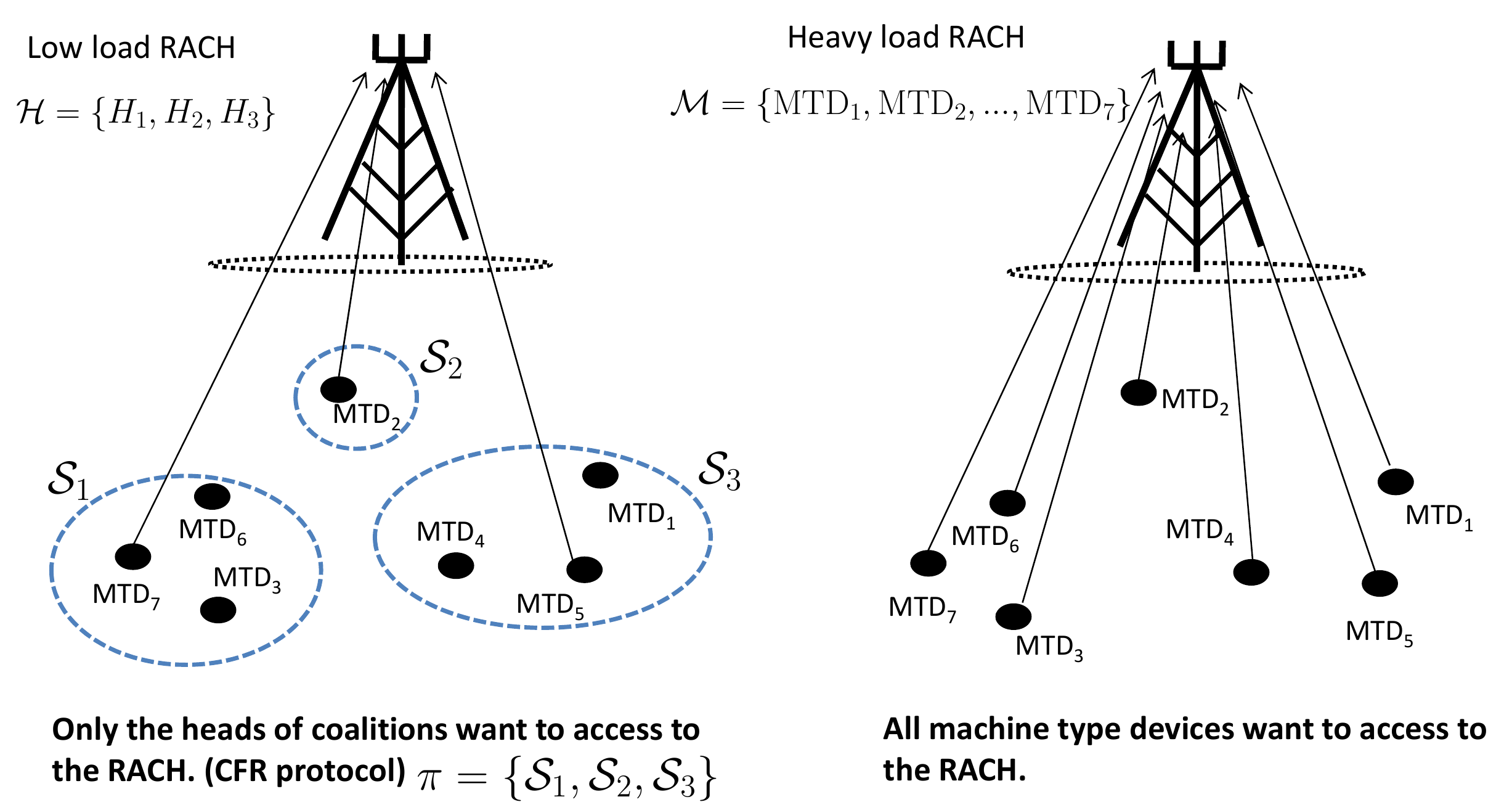}
\caption{An illustration example of stochastic coalition formation and traditional RACH for $\mathcal{M}=7$.}
\label{SModel}
\end{figure}

Fig. \ref{SModel} shows the noncooperative and cooperative random access model of 7 MTDs for M2M applications. Under a noncooperative random access model, all 7 MTDs individually send their access request on their cellular links to BS. In this case, the RACH model is overloaded by all 7 MTDs.  However, under a cooperative random access model, MTDs, which are within the coverage of the short-range M2M links, form a cooperative coalition. Following Fig. \ref{SModel}, MTD $3$, MTD $6$ and MTD $7$ form coalition $\mathcal{S}_1$, MTD $2$ forms coalition $\mathcal{S}_2$ and  MTD $1$, MTD $4$ and MTD $5$ form coalition $\mathcal{S}_3$. Thus, during time slot $t$, the 7 MTDs form cooperative coalitions $\Pi=\{\mathcal{S}_1,\mathcal{S}_2,\mathcal{S}_3\}$. In this case, the RACH load is affected just by the three MTHs of the formed coalitions which is much less than 7 MTDs in noncooperative model.
\subsection{Optimal Solution for Cooperative Random Access in M2M Communications}
We present an optimization formulation for the cooperative random access in M2M communications. {If global network information is available, then the optimal coalitions can be computed centrally at the base station. Given a network partition $\Pi$ composed of $N$ coalitions, then, let $k_{mn}$ be a binary variable such that $k_{mn} = 1$ if $\text{MTD}_m$ belongs to the coalition $\mathcal{S}_n$ otherwise $k_{mn}=0$. The centralized optimization problem will be:}
{\begin{align}
& \underset{\left\{
\substack{\Pi,[k_{mn}]_{M\times N}}
\right\}}
\min
\sum_{m} \sum_{t=t_0}^{\infty} \delta^{t-t_0} Q_{m,t}, \label{altruprob}\\
& \sum_{t=t_0}^{\infty} \delta^{t-t_0} \big(w_m E_{LR}(\Pi)+E_{SR}(\mathcal{S}_i)\big) \leq E_{\max} \text{, } \forall m \in \mathcal{M}, \label{altruprob_c1}\\
& k_{mn}\times k_{rn}\frac{B}{T} \leq R_{rm} \text{, }  \forall \mathcal{S}_n \in \Pi, \label{altruprob_c2}\\
&\sum_{n} k_{mn}=1 \text{ , } \forall m \in \mathcal{M}, \label{altruprob_c3}\\
&\sum_{m\in \mathcal{M}}\sum_{n \in N} k_{mn}=M, \label{altruprob_c4}\\
&k_{mn} \in \{0,1\} \text{, }  \forall m \in \mathcal{M},\forall \mathcal{S}_i \in \Pi. \label{altruprob_c5}
\end{align}}

{The objective function in (\ref{altruprob}) minimizes the discounted sum of the queue lengths of the MTDs over time. (\ref{altruprob_c1}) guarantees that the discounted sum of per-request energy consumption of all MTDs is less than its predefined maximum value. (\ref{altruprob_c2}) indicates that the bit rate of the M2M link between two MTDs in each coalition should be high enough to allow the transmission of one $B$-bit packet between them. (\ref{altruprob_c3}) shows that each MTD can only be on one coalition. (\ref{altruprob_c4}) guarantees that all MTDs are considered in the coalition formation. (\ref{altruprob_c5}) shows that $k_{mn}$ is a binary. (\ref{altruprob}) is a combinatorial optimization problem with an exponential search space. The reason is that the number of all partitions given by a value known as the Bell number which grows exponentially
with the number of MTDs in the coalitions~\cite{sandholm1999coalition} and~\cite{sen2000searching}. (\ref{altruprob}) can be solved by potentially deterministic exhaustive search algorithm\cite{sandholm1999coalition}. However, due to size of our problem, we apply genetic algorithms which are a class of stochastic search algorithms that have been used widely to solve large-scale NP-complete combinatorial optimization problems including searching for optimal coalition structures\cite{sen2000searching}.}
\section{Stochastic Game-theoretic Model for M2M Cooperation}
\label{Sec:Gam-Model}
In this section, we model the proposed cooperative random access model for M2M communication over the RACH using cooperative game theory~\cite{dawy2015towards}. First, we focus on the cooperation of MTDs during each time slot given the stochastic queue lengths of MTDs and we model it using a \emph{coalition game in stochastic characteristic function form} (CGSC)~\cite{granot1977cooperative}. The goal of this CGSC is to find the formed coalitions during each time slot. Then, for modeling the cooperation of MTDs during different time slots, we use the framework of \emph{stochastic coalitional games} (SCGs)~\cite{kaluski2002n}. An SCG captures how the formation of cooperative coalitions can change due to stochastic factors, such as random arrivals. Thus, in our model, an SCG can be seen as a repeated game with one of its stages being a CGSC that is played by MTDs. Finally, we propose an algorithm to solve the SCG in stochastic characteristic function form and find a stable partition.
\subsection{Coalitional game formulation {during} one time slot}
In each time slot, a coalitional game in stochastic characteristic function form is uniquely defined by the pair ($\mathcal{M},v_t$), where the set $\mathcal{M}$ of players is the set of MTDs and $v_t:\Pi_t\rightarrow \mathbb{R}^{|\mathcal{S}_i|}$ is a stochastic value that reflects the utilities achieved by the coalitions formed at a given time slot $t$. This game is different from a classical coalitional game such as in~\cite{aumann1992handbook} in that the coalitional value is stochastic. First, from (\ref{exact success rate}) and (\ref{coalition queue}), we can see that queue length of each coalition, which depends on the packet success rate over RACH, is affected by other coalitions. Thus, the cooperation model in each time slot can be mapped to a coalitional game in partition form which can be significantly challenging to solve when the value is stochastic~\cite{young2014handbook}. To over come this challenge, the cooperating MTDs will assume that all of the other MTDs at each time slot want to access the RACH and do not form any coalition. Such an assumption maps to a conservative strategy in which the MTDs in each coalition assume the worst-case collision rate from other MTDs. This is in line with existing works such as the jamming games in \cite{xiao2015user}, in which it is assumed that all other players jam a coalition in a cognitive radio network. Consequently, in this worst case scenario, the probability of the packet success rate over RACH, $d_{\mathcal{S}_i,t}=1$, is given by:
\begin{equation}
\textrm{Pr}(d_{\mathcal{S}_i,t}=1)=P_s^{\mathcal{H}}=\frac{1}{\mu}(1-\frac{1}{\mu})^{M-|{\mathcal{S}_i}|}.
\label{approx success rate}
\end{equation}

In this case, we can compute the value of each coalition independently of the coalition decisions of other MTDs. Now, the value of a coalition $S_i$ will depend only on its members, $v_t:\mathcal{S}_i\rightarrow \mathbb{R}^{|\mathcal{S}_i|} $, and, thus, we have a game in characteristic function form~\cite{young2014handbook}.

Given the dynamic changes in queues , $Q_t$, and the probability of the packet success rate, $P_s^{\mathcal{H}_t}$, over the RACH, the value of the coalitions will randomly change over time. Thus, during each time slot $t$, the stochastic value of this game is a random variable which consists of two components ($u_d$, $u_r$)~\cite{granot1977cooperative}: a deterministic value ($u_d$) gained in the current time slot and a random value ($u_r$) that captures the prospective gains in future time slots.

\subsubsection{Deterministic value for the current time slot}
The deterministic value for the current time slot of a coalition can be directly computed by the members of a coalition since it is a certain outcome. This deterministic component is not related to the random value changes that will happen in future time slots. We consider two cooperation schemes for the proposed RACH coalitional game: \emph{altruistic cooperation}, in which MTDs cooperate to maximize the value of their coalition or \emph{selfish cooperation} in which MTDs may cooperate to increase their individual payoffs.

\textbf{Altruistic cooperation}: altruistic MTDs seek to decrease the overall group queue length and group energy consumption of their formed coalition. In such a scenario, MTDs are mainly concerned with the overall gain of the coalition rather than their individual payoffs. In this case, the deterministic current time slot value of the coalition $u_{d}({\mathcal{S}_i},Q_{\mathcal{S}_i,t})\in \mathbb{R}$ is a real value that is equal to the payoff of each machine MTD $m$ in the coalition $\mathcal{S}_i$. For each MTD $m$ in $\mathcal{S}_i$, the payoff function will be equal to the value function::
\begin{equation}
u_{d}({\mathcal{S}_i},Q_{\mathcal{S}_i,t})=-\alpha_m |Q_{\mathcal{S}_i,t}|-\beta_m \left(\frac{\bar{E}_{LR}^{H}}{|\mathcal{S}_i|}+E_{SR}\right)-\gamma_m |\mathcal{S}_i|,
\label{altruistic coalition}
\end{equation}
where $\gamma_m$ is a unit cost parameter for MTD $m$. $\bar{E}_{TL}^{H}$ and  $Q_{\mathcal{S}_i,t}$ are respectively given by (\ref{average energy}) and (\ref{coalition queue}), when we use (\ref{approx success rate}) as the probability of the packet success rate over the RACH which is same as the probability of departure from the queue.

\textbf{Selfish cooperation}: each selfish MTD will seek to decrease its individual queue length and energy consumption in the formed coalition, while disregarding the overall social welfare of the entire coalition. In this case, the deterministic current time slot value of the coalition is no longer a function over the real line. Instead, the deterministic current time slot value of a coalition is a set of payoff vectors, $\mathcal{U}_{d}({\mathcal{S}_i},Q_{\mathcal{S}_i,t})\subseteq \mathbb{R}^{|\mathcal{S}_i|}$, where each element $u_{d,m}({\mathcal{S}_i},Q_{\mathcal{S}_i,t})$ of any vector $\boldsymbol{u}_{d}\in \mathcal{U}_{d}$ represents the payoff of each MTD $m$ in the coalition $\mathcal{S}_i$, which is given by:
\begin{equation}
u_{d,m}({\mathcal{S}_i},Q_{\mathcal{S}_i,t})=-\alpha_m |Q_{m,t}|-\beta_m (w_m^e \bar{E}_{LR}^{H}+E_{SR})-\gamma_m |\mathcal{S}_i|,
\label{selfish coalition}
\end{equation}
where $\bar{E}_{TL}^{H}$ and  $Q_{m,t}$ are respectively given by (\ref{average energy}) and (\ref{member coalition queue}) when we use (\ref{approx success rate}) as a probability of the RACH packet success rate (which is the same as the probability of departure from the queue) and $w_m^e$ is the fairness weight for time of being head.
\subsubsection{Random value for future time slots}
To capture the random component of the value of each coalition, we consider discounted future rewards for each MTD. For each MTD $m$ at time slot $t$, a discounted future reward is the sum of future payoffs during the following time slots, $t+1,t+2,...$, which are discounted by a constant factor~\cite{shoham2008multiagent}. If we consider the same discount factor $\delta$, $0\leq \delta \leq 1$, for all MTDs, then a higher $\delta$ will reflect MTDs that are farsighted and, thus, more interested in the payoff of the next time slots rather than the current one.  Thus, the random values $u_r(\mathcal{S}_i,Q_{\mathcal{S}_i,t})$ for the altruistic scheme, and $u_{r,m}(\mathcal{S}_i,Q_{\mathcal{S}_i,t})$ for the selfish scheme that capture the payoff during future time slots, will be given by:
\begin{align}
& u_r(\mathcal{S}_i,Q_{\mathcal{S}_i,t})=\sum_{n=t+1}^{n_{\delta}}\delta^{n-t} \bar{u}_d(\mathcal{S}_i,Q_{\mathcal{S}_i,t\rightarrow n})\nonumber \\
& u_{r,m}(\mathcal{S}_i,Q_{\mathcal{S}_i,t})=\sum_{n=t+1}^{n_{\delta}}\delta^{n-t} \bar{u}_{d,m}(\mathcal{S}_i,Q_{\mathcal{S}_i,t\rightarrow n})
\label{Value_fun}
\end{align}

Here, $n_{\delta}$ is a discrete value where $\delta^{m}\simeq 0$ for $n_{\delta} \leq m$. $u_d(\mathcal{S}_i,Q_{\mathcal{S}_i,t\rightarrow n})$ and $u_{d,m}(\mathcal{S}_i,Q_{\mathcal{S}_i,t\rightarrow n})$ are stochastic processes which stochastically change from current time slot $t$ to the future time slot $n$. $\bar{u}_d(\mathcal{S}_i,Q_{\mathcal{S}_i,t\rightarrow n})$ and $\bar{u}_{d,m}(\mathcal{S}_i,Q_{\mathcal{S}_i,t\rightarrow n})$ are the expectation of random changes of $u_d(\mathcal{S}_i,Q_{\mathcal{S}_i,t\rightarrow n})$ and $u_{d,m}(\mathcal{S}_i,Q_{\mathcal{S}_i,t\rightarrow n})$  from time slot $t$ to the time slot $n$. These values depend on the change in the queue during $n-t$ time slots, $Q_{\mathcal{S}_i,t\rightarrow n}$. The stochastic change of the queue in one time slot, $Q_{\mathcal{S}_i,t\rightarrow t+1}$, is modeled by the probabilistic model shown in Fig. \ref{Probability_model}. Since the arrivals and departures of requests from each coalition are independent event, $\text{Pr}(d_{i,t}=i,a_{i,t}=j)=\text{Pr}(d_{i,t}=i)\times \text{Pr}(a_{i,t}=j)$ where $\text{Pr}(d_{i,t}=i)$ and $\text{Pr}(a_{i,t}=j)$ are given by (\ref{arrival_prob}) and (\ref{exact success rate}). In general,  $\bar{u}_d(\mathcal{S}_i,Q_{\mathcal{S}_i,t\rightarrow n})$ is:
\begin{figure}[t]
\centering
\includegraphics[width=8.0cm,height=1.25in]{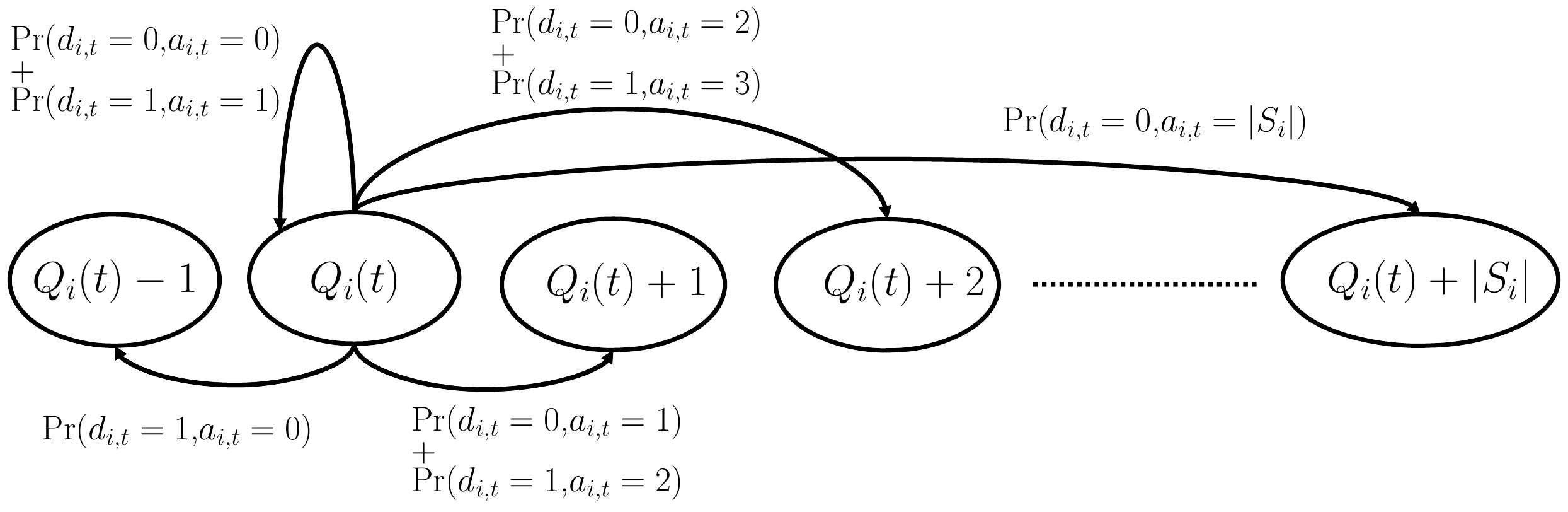}
\caption{The transition probability model for queue changes during one time slot.}
\label{Probability_model}
\end{figure}
\begin{equation*}
\bar{u}_d(\mathcal{S}_i,Q_{\mathcal{S}_i,t\rightarrow n})=\Delta_{t}^{n} u_d({\mathcal{S}_i},Q_{\mathcal{S}_i,n})\text{, for the altruistic scheme, and}
\end{equation*}
\begin{equation}
\bar{u}_{d,m}(\mathcal{S}_i,Q_{\mathcal{S}_i,t\rightarrow n})=\Delta_{t}^{n} u_{d,m}({\mathcal{S}_i},Q_{\mathcal{S}_i,n})\text{, for the selfish scheme.}
\label{U_trans}
\end{equation}
Here, $u_d({\mathcal{S}_i},Q_{\mathcal{S}_i,n})$ and $u_{d,m}({\mathcal{S}_i},Q_{\mathcal{S}_i,n})$ are given by (\ref{altruistic coalition}) and (\ref{selfish coalition}), respectively, and $\Delta_{t}^{n}$ is the probability of transition from $Q_{\mathcal{S}_i,t}$ to $Q_{\mathcal{S}_i,n}$ during $n-t$ time slots. This probability depends on the $(n-t)$-hop path from $Q_{\mathcal{S}_i,t}$ to $Q_{\mathcal{S}_i,n}$ in Fig.~\ref{Probability_model}. This path reflects how the queue of coalition $\mathcal{S}_i$ changes during $n-t$ time slots, $Q_{\mathcal{S}_i,t}\rightarrow Q_{\mathcal{S}_i,n}$. The probability of change in the queue is equal to the sum of probabilities over all paths from $Q_{\mathcal{S}_i,t}$ to $Q_{\mathcal{S}_i,n}$:
\begin{equation}
\Delta_{t}^{n}=\sum_{\text{all paths: }Q_{\mathcal{S}_i,t}\rightarrow Q_{\mathcal{S}_i,n}} \textrm{Pr}({Q_{\mathcal{S}_i,t}\rightarrow Q_{\mathcal{S}_i,n}}).
\label{Cal_delta}
\end{equation}
From Fig.~\ref{Probability_model}, we can see that there is a one-hop path between $Q_{\mathcal{S}_i,t}$ and $Q_{\mathcal{S}_i,t+1}$. This path captures the stochastic change of a queue during one time slot. This one-hop path reflect the event during which queue $Q_{\mathcal{S}_i,t}$ changes to $Q_{\mathcal{S}_i,t}+j$ where $j=-1,0,...,|\mathcal{S}_i|$. The probability that the queue changes due to a one-hop path, $p(Q_{\mathcal{S}_i,t}\rightarrow Q_{\mathcal{S}_i,t}+j)$, can be found in Fig. \ref{Probability_model}. For $\Delta_{t}^{t+2}$, the number of paths becomes greater than one. For example, if $Q_{\mathcal{S}_i,t+2}=Q_{\mathcal{S}_i,t}+3$ then there are $|\mathcal{S}_i|+2$ two-hop paths during 2 time slots to change the queue length from $Q_{\mathcal{S}_i,t}$ to $Q_{\mathcal{S}_i,t}+3$, which are $Q_{\mathcal{S}_i,t}\rightarrow Q_{\mathcal{S}_i,t}+i\rightarrow Q_{\mathcal{S}_i,t}+3$ for $i=-1,0,...,|\mathcal{S}_i|$. Then, the probability of each path is calculated by multiplying the probability of two hops of Fig.~\ref{Probability_model}.

After computing the deterministic value for current time slot and the random value for future time slots, we can explicitly define the value of each coalition or payoff function of each MTD. For instance, when MTDs are altruistic, the value of coalition $\mathcal{S}_i$ is equal to the payoff function of each MTD in the coalition as follows:
\begin{equation}
v^a(\mathcal{S}_i,Q_{\mathcal{S}_i,t})=u_{d}({\mathcal{S}_i},Q_{\mathcal{S}_i,t_0})+u_{r}({\mathcal{S}_i},Q_{\mathcal{S}_i,t_0}).
\label{altruistic Value_fun}
\end{equation}

When MTDs are altruistic, the value of each coalition $\mathcal{S}_i$  will be a real value that is equal to the benefit achieved by each individual MTD. If the MTDs are selfish, then the deterministic current time slot value and random value for future time slots will be, respectively, a set of payoff vectors $\mathcal{U}_{d}({\mathcal{S}_i},Q_{\mathcal{S}_i,t})$ and $\mathcal{U}_{r}({\mathcal{S}_i},Q_{\mathcal{S}_i,t})\subseteq\mathbb{R}^{|\mathcal{S}_i|}$. Thus, the value of a coalition, which is the sum of the deterministic current time slot value and random value for future time slots, becomes a set of vectors, $\mathcal{V}^s({\mathcal{S}_i},Q_{\mathcal{S}_i,t})\subseteq \mathbb{R}^{|\mathcal{S}_i|}$ with each element $v_{m}^s(\mathcal{S}_i,Q_{\mathcal{S}_i,t})$ of any vector $\boldsymbol{v}^s \in \mathcal{V}^s$ is the payoff function of each MTD in the coalition. The payoff function of a given MTD $m$ in coalition $\mathcal{S}_i$ is given by:
\begin{equation}
v_{m}^s(\mathcal{S}_i,Q_{\mathcal{S}_i,t})=u_{d,m}({\mathcal{S}_i},Q_{\mathcal{S}_i,t_0})+u_{r,m}({\mathcal{S}_i},Q_{\mathcal{S}_i,t_0}).
\label{selfish Value_fun}
\end{equation}

Following (\ref{altruistic Value_fun}) and (\ref{selfish Value_fun}), in the CGSC, the value of a coalition and the payoff of each MTD will depend on the formed coalition and its queue. Since the queue cannot be arbitrary shared among members of coalition, the CGSC is a game with non-transferable utility (NTU)~\cite{young2014handbook}.
\subsection{Random access as an $M$-person stochastic coalition game}
Next, we model the cooperation of MTDs during different time slots as an SCG~\cite{kaluski2002n}. In each stage of this SCG, the cooperation of MTDs is modeled using a CGSC~\cite{granot1977cooperative}. The players in the SGC are also the MTDs in the set $\mathcal{M}$. The SCG is played in such a way that, at each time slot $t$, the MTDs decide on their membership in a formed coalition $\Pi(t)$. Their decisions depend on the random conditions of the game. The state of the SCG at a slot $t$ is denoted by $h_t=(\Pi(t),\mathcal{Q}(t))$. This state is a two-dimensional random variable with discrete finite states which means MTDs form coalitions $\Pi(t)$ while $\mathcal{Q}(t)$ is the set of queue of the coalitions. $\mathcal{Q}(t)$ is stochastically changing at each time slot $t$. Consequently, the RACH \emph{$M$-person stochastic coalition game in stochastic characteristic function} (MSCF) is uniquely defined by the triplet $(\mathcal{M},v_t,h_t)$. Thus, the MTDs choose the actions that lead to forming disjoint coalitions and each MTD receives a stochastic payoff $v_t$ at each time slot $t$ which is given by (\ref{altruistic Value_fun}) or (\ref{selfish Value_fun}).

The formed MTD coalitions may randomly change during different time slots, because the arrival and departure processes of the coalitions' queues can stochastically change. Therefore, at each stage of an MSCF, the MTDs form the most suitable coalitions depending on the queues. Let $h_t=(\Pi(t),\mathcal{Q}(t))$ be the state of the coalition formation process at time slot $t$. As shown in ~\cite{aumann1992handbook}, the process of coalition formation is a stochastic process $\zeta$ which starts from an initial state $h_0$ and moves to another state $h_t$ following the stochastic changes in the queues of the coalitions when going form the time slot $0$ to the time slot $t$. Our goal is to find a stable state of $\zeta$ which is essentially a stable partition, $\Pi^*$, of $\mathcal{M}$. Consequently, we can find how coalitions are formed by MTDs given the stochastic changes of their queues.

First, we define the moves or decisions that are going to be used in our proposed algorithm. Consider a state $h_t=(\Pi(t),\mathcal{Q}(t))$ and a coalition $\mathcal{S}_i$, then, we make the following definition:
\begin{definition}
\textnormal{Let $\mathcal{F}_{\mathcal{S}_i}(h_t)$ be the set of states achievable by a \emph{one-step coalitional move} (by $\mathcal{S}_i$) which changes the coalition formation, $\Pi(t)$, when an $M$-person stochastic coalition game in stochastic characteristic function form is in state $h_t$.}
\end{definition}
There are three types of moves for each MTD $m$ at each time slot $t$. The first is $\mathcal{S}_i$, which means that the members of $\mathcal{S}_i$ do not change their coalition. The second type of moves for each MTD $m$ is $\mathfrak{C}_m=\{\mathcal{C}_1,\mathcal{C}_2,...,\mathcal{C}_{|\mathcal{S}_i|-1}\}$, where $\mathcal{C}_k$ is a $k$-person coalition for MTD $m$ that consists of MTD $m$ and $k-1$ members from $\mathcal{S}_i$. $\mathfrak{C}_m$ is the set of all possible $k$-person coalitions that MTD $m \in \mathcal{S}_i$ can form with $k < |\mathcal{S}_i|$ other members in the coalition $\mathcal{S}_i$. The total number of moves in $\mathfrak{C}_m$ is $|\mathfrak{C}_m|=\sum_{k=1}^{|\mathcal{S}_i|-1} \mathcal{C}_k$ where $|\mathcal{C}_k|=\frac{(|\mathcal{S}_i|-1)!}{(k-1)!(|\mathcal{S}_i|-k-2)!}$. The last type of move occurs when MTD $m$ in the coalition $\mathcal{S}_i$ asks other coalition $\mathcal{S}_j$ to form a new larger coalition which is $\mathcal{S}_i\cup\ \mathcal{S}_j$. Consequently, the total of one-step move for each MTD $m$ in the coalition $\mathcal{S}_i$ given by:
\begin{equation}
\mathcal{F}_{\mathcal{S}_i}^m(h_t)=\mathcal{S}_i\cup\ \mathfrak{C}_m\cup\ \{{\mathcal{S}_{ij}|\forall j\neq i}\}.
\end{equation}
Therefore, a one-step coalitional move is equal to the union of the one-step moves of all of the MTDs in $\mathcal{S}_i$ denoted by $\mathcal{F}_{\mathcal{S}_i}(h_t)=\{\cup _{m\in \mathcal{S}_i}\mathcal{F}_{\mathcal{S}_i}^m(h_t)\}$. After determining all possible one-step coalitional moves, we can now define the concept of a profitable move~\cite{konishi2003coalition}:
\begin{definition}
\textnormal{$\mathcal{S}_i$ has a \emph{(weakly) profitable move} from $\Pi_t^1$ (under $\zeta$) if there is $\Pi_t^2\in \mathcal{F}_{\mathcal{S}_i}(h_t)$ (with $\Pi^2(t)\neq \Pi^1(t)$) such that $v_m(\Pi^2(t),\zeta)\geq v_m(\Pi^1(t),\zeta)$ for all $m \in \mathcal{S}_i$. $\mathcal{S}_i$ has a strictly profitable move from $\Pi^1(t)$ if there is $\Pi^2(t)\in \mathcal{F}_{\mathcal{S}_i}(h_t)$ such that $v_m(\Pi^2(t),\zeta)>v_m(\Pi^1(t),\zeta)$  for all $m \in \mathcal{S}_i$.}
\end{definition}

By using the notion of a profitable move, we will propose an $M$-person stochastic coalition formation algorithm that can be used to enable the MTDs to perform distributed coalition formation under a stochastic characteristic function form. The phases of algorithm are:

\textit{Initial Phase}: At each time slot $t_0$, the initial partition can be a singleton network partition: $\Pi(t_0) = \{\{1\}, \{2\},..., \{M\}\}$ or a grand coalition:$\Pi(t_0)=\{1,2,...,M\}$.

\textit{Sequential Phase}: The stochastic coalition formation algorithm keeps on iterating over all the MTDs in the network until all MTDs decide to stay in their current coalition, which indicates that the
algorithm has converged to a final stable network partition $\Pi^*(t_0)$. In each iteration, each MTD $m$ member of a the coalition $\mathcal{S}_i$ follows its most preferable one-step coalitional move which is in $\mathcal{F}_{\mathcal{S}_i}^m(h_t)$.

A summary of the stochastic coalition formation algorithm is presented in Table~\ref{SIDG}.
\begin{table}[!t]
  \centering
  \caption{
    \vspace*{-0em}Stochastic Coalition Formation Algorithm for M2M Communication}\vspace*{-0em}
    \begin{tabular}{p{8cm}}
      \hline \vspace*{-0em}
      \textbf{Inputs:}\,\,$\mathcal{M},\Pi(t_0),\mathcal{Q}(t_0)$\\
\hspace*{1em}\textit{Initialize:}   \vspace*{0em}

Set initial state and discount factor as $h_{t_0}=\{\Pi(t_0),\mathcal{Q}(t_0)\}$ and $\delta$, respectively.\vspace*{-0cm}

Find the discrete value $n_{\delta}$ where $\delta^{m}\simeq 0$ for $n_{\delta} \leq m$.

\hspace*{0em} \textit{Stage 1:}
\begin{itemize}\vspace*{-0em}
\item[] \hspace*{0em}(a) For each coalition $\mathcal{S}_i$, calculate stochastic change of the queue, $Q_{\mathcal{S}_i,t_0\rightarrow t_0+n_{\delta}}$, until time slot $t_0+n_\delta$,  using the transition probability model in Fig. 2.
\item[] \hspace*{0em}(b) Determine the value functions of MTDs in the partition $\Pi(t_0)$, using (\ref{Value_fun}).
\item[] \hspace*{0em}(c) For each MTD $m\in\mathcal{S}_i$, find $\mathfrak{C}_m$ including all possible $k$-person coalitions.
    \item[] \hspace*{0em}(d) For each MTD $m\in\mathcal{S}_i$, find all possible coalitions such as $\mathcal{S}_j$ that forming $\mathcal{S}_i\cup\ \mathcal{S}_j$ is a profitable move.
\end{itemize}\vspace*{-0cm}
\hspace*{0em}\textit{Stage 2:}
\begin{itemize}\vspace*{-0em}
\item[]  \hspace*{0em}(a) For each coalition $\mathcal{S}_k\in \mathcal{F}_{\mathcal{S}_i}^m(h_t)$, calculate stochastic change of the queue, $Q_{\mathcal{S}_k,t\rightarrow t+n_{\delta}}$, until time slot $t_0+n_\delta$, using the transition probability model in Fig. 2.
\item[]  \hspace*{0em}(b) Do the most preferable one-step coalitional move in $\mathcal{F}_{\mathcal{S}_i}(h_t)$.
\end{itemize}\vspace*{-0cm}
\hspace*{0em}\textit{Stage 3:}
\hspace*{0em}\textbf{while} $\Pi(t_0)$ changes for two consecutive iterations\\
\hspace*{0em}\textit{repeat Stage 1 to Stage 2}\vspace*{0em}\\
\hspace*{0em}\textbf{Output:}\,\,Stably formed coalition: $\Pi^*$\vspace*{0em}\\
   \hline
    \end{tabular}\label{tab:algo}\vspace{-0.5cm}
\label{SIDG}
\end{table}
Next, we analyze the convergence and stability of the proposed algorithm. First, we define a deterministic process of coalition formation as follow~\cite{konishi2003coalition}:
\begin{definition}
\textnormal{For all formed coalitions $\Pi_t^1$ and $\Pi_t^2$, a process of coalition formation is said to be \emph{deterministic} if the transition probability from coalition $\Pi_t^1$ to coalition $\Pi_t^2$ in each time slot $t$, $p(\Pi_t^1,\Pi_t^2)$ is in $\{0,1\}$.}
\end{definition}
\begin{lemma}
\textnormal{The proposed $M$-person stochastic coalition formation algorithm in Table~\ref{SIDG} is a deterministic process of coalition formation.}
\label{Theor_Determ}
\end{lemma}
\begin{proof}
During the coalition formation process in Table~\ref{SIDG}, when coalition $\Pi^1_t$ is formed at a given step, if there is at least one profitable move for MTDs, the algorithm in Table~\ref{SIDG} will go to coalition $\Pi^2_t$ in the next step with probability one. If there is no profitable move for MTDs, the probability to change the formed coalition from $\Pi^1_t$ to $\Pi^2_t$ is zero. Therefore, the algorithm in Table~\ref{SIDG} deterministically changes the coalition formation process.
\end{proof}
\begin{definition}
\textnormal{Under a stochastic change of the coalition value, a formed partition $\Pi^*$ is said to be \emph{stable}, if no MTD $m$ in the coalition $\mathcal{S}_i^*$ nor any group of MTDs can form another coalition. For a stable coalition, $\Pi^*$, the following conditions are satisfied for all $\mathcal{S}_i^*\text{ and } \mathcal{S}_j^*\in \Pi^*$:
\begin{equation*}
v_m(\mathcal{S}_i^*,Q_{\mathcal{S}_i,t})\geq v_m(\mathcal{C}_k,Q_{\mathcal{C}_k,t})\text{, } k=1,2,...,|\mathcal{S}_i^*|-1,
\end{equation*}
\begin{equation}
v_m(\mathcal{S}_i^*,Q_{\mathcal{S}_i,t})\geq
v_m(\mathcal{S}_i^*\cup \mathcal{S}_j^*,Q_{\mathcal{S}_i,t}+Q_{\mathcal{S}_j,t}).
\label{stable_cond}
\end{equation}
}
\label{Def_stability}
\end{definition}

Based on (\ref{stable_cond}), in a stable partition, $\Pi^*$, no MTD $m$ in $\mathcal{S}_i^*$ can benefit by forming any $k$-person coalition that is different from the current coalition $\mathcal{S}_i^*$ nor by enabling its coalition to form a new, bigger coalition with other coalition $\mathcal{S}_j^*$. Thus, in a stable partition, each MTD $m$ prefers to stay in its current coalition. As shown in~\cite[Theorem 4.1]{konishi2003coalition}, for any deterministic process of coalition formation in characteristic function form, there exists a discount factor $\delta^*\in(0,1)$ such that for any collection of discount factors in $(\delta^*,1)$, that deterministic process converges to the unique limit which is the coalitions that will effectively form. According to~\cite[Theorem 4.1]{konishi2003coalition} and Lemma~\ref{Theor_Determ}, we can state the following result for our proposed algorithm in Table I.
\begin{theorem}
\textnormal{For the proposed deterministic $M$-person stochastic coalition formation algorithm in Table I, there exists a discount factor $\delta^*$ such that for any collection of discount factors in $(\delta^*,1)$, the unique limit of algorithm in Table I will be $\Pi^*$ with $v(\mathcal{S}_i^*,Q_{\mathcal{S}_i,t})$ . $\delta^*$ can be found from the following conditions:
\begin{align}
& \sum_{n=t+1}^{\infty}\delta^{n-t}\big(\Delta_{t}^{n} u_d({\mathcal{S}_i^*},Q_{\mathcal{S}_i^*,n})-\Delta_{t}^{n} u_d({\mathcal{C}_k},Q_{\mathcal{C}_k,n})\big)
\geq \nonumber \\
&u_{d,m}({\mathcal{C}_k},Q_{\mathcal{C}_k,t})-u_{d,m}({\mathcal{S}_i^*},Q_{\mathcal{S}_i^*,t}),\nonumber \\
& \forall \mathcal{S}_i^* \in \Pi^*, k=1,2,...,|\mathcal{S}_i^*|-1,\text{ and,}\nonumber \\
&\sum_{n=t+1}^{\infty}\delta^{n-t}\big(\Delta_{t}^{n} u_d({\mathcal{S}_i^*},Q_{\mathcal{S}_i^*,n})-\Delta_{t}^{n} u_d(\mathcal{S}_i^*\cup\ \mathcal{S}_j^*,Q_{\mathcal{S}_i^*\cup\ \mathcal{S}_j^*,n})\big)
\geq \nonumber \\
&u_{d,m}(\mathcal{S}_i^*\cup\ \mathcal{S}_j^*,Q_{\mathcal{S}_i^*\cup\ \mathcal{S}_j^*,t})-
u_{d,m}({\mathcal{S}_i^*},Q_{\mathcal{S}_i^*,t}),\nonumber \\
& \forall \mathcal{S}_j^* \in \Pi^*, \mathcal{S}_j^*\neq\mathcal{S}_i^*.
\end{align}
}
\label{Theor_band_delta}
\end{theorem}

\begin{proof}
 See Appendix A.
\end{proof}
Following Lemma \ref{Theor_Determ}, Theorem \ref{Theor_band_delta}, and Definition \ref{Def_stability}, we can see that there exists $\delta^* \in (0,1)$ such that for all discount factors in $(\delta^*,1)$, the proposed algorithm in Table~\ref{SIDG} will converge to a stable partition $\Pi^*$.
\subsection{Complexity and Stability Analysis}
The complexity of the algorithm in Table I lies in the complexity of the preferable one-step coalitional movement. Therefore, a good measure of complexity will be the number of possible (per coalition) movements. For a given network structure, one possible move for a coalition $S_i$ is to form a new larger coalition $\mathcal{S}_i\cup \mathcal{S}_j$. This move implies that a given coalition will try to merge with other coalitions while ensuring that the bit rate between the members of the merged coalition is equal to or larger than the desirable rate, i.e., $\frac{B}{T}\leq R_{im}$. In the most complex case, the total number of merge attempts will be $\frac{M(M-1)}{2}$. The second type of coalition moves occurs when the MTDs in a given coalition $\mathcal{S}_i$ want to split and  form a $k$-person coalition out of coalition $\mathcal{S}_i$ where $k<|\mathcal{S}_i|$. The total number of split moves for each MTD $m$ is $|\mathfrak{C}_m|$. The most complex case for splitting occurs when the network forms an $M$-person grand coalition. The total number of split attempts will be $\sum_{m=1}^{M}|\mathfrak{C}_m|$. In terms of computations, for each movement, the value function should be calculated under stochastic changes. This calculation depends on the queue length in the coalition and the value of $\delta$. A larger value of $\delta$ leads to the need for computing additional terms in (14) and (16). To show the effect of $\delta$ on the number of terms of the series in (14) and (16), for each value of $\delta$, we define a discrete value $n_{\delta}$ where $\delta^{m}\simeq 0$ for $n_{\delta} \leq m$. This means that value function should be calculated until next $n_{\delta}$ time slots. Following (14) and (16), in the worst case, the complexity of calculating the value function is $\big(\min\{M,K\}\big)^{n_{\delta}}$ where $M$ is the maximum size of grand coalition and $K$ is maximum available buffer size. Consequently, the complexity of the proposed algorithm in Table I is $O\big(\big( \sum_{m=1}^{M}|\mathfrak{C}_m|+\frac{M(M-1)}{2} \big) \times \big(\min\{M,K\}\big)^{n_{\delta}}\big)$.

In practice, the merge process requires a significantly lower number of attempts because the MTDs can only attempt to merge with coalitions within their communication range that also satisfy the bit rate requirement. Since the size of a coalition is typically small and limited due to the cost for cooperation, the split moves will be limited to finding all possible partitions for small sets. Consequently, we will seldom encounter an $M-$sized coalition. In fact, most formed coalitions will be of size $K$ that is much smaller than $M$.

\section{Simulation Results and Analysis}
\label{Sec:Performance}
For our simulations, we consider a single cell served by one BS that is located at the center of a $400$~m~$\times$~$400$~m square area. We consider two distribution models of the MTDs around the BS: (A) uniform and (B) cluster-based distributions. In the uniform distribution model, the MTDs randomly distributed across square area. {In a cluster-based distribution, we consider cluster centers with density of $5\times 10^{-5}$ cluster/$\text{m}^2$ whose locations result from the realization of a Poisson point process. In each cluster, the locations of the MTDs result from the realization of a normal distribution around the cluster center. The number of MTDs per cluster is determined by a Poisson distribution whose mean value is equal to the ratio of the number of MTDs to the number of cluster centers in each realization.} The time slot for complete RA request is set to $T=1$~millisecond. The number of available preambles for RACH is $\mu=256$. The maximum power consumption on each RB for sending data over the cellular link to the BS is set to $25$~dBm. The maximum size of each queue is $K=30$. We assume that the maximum power consumption used by MTDs for sending requests to a coalition head is $5$~dBm on each RB~\cite{ho2012energy} during each time slot. The bandwidth of each resource block is 15 kHz. We consider a 2 GHz carrier frequency. The noise power spectral density $N_0$ is −170 dBm per Hz. We consider a path loss exponent of 2.5 and a Rayleigh fading with mean 1 for the channel model. Moreover, the length of each packet is set to 50 bits. We consider that the probability of sending a request for each MTD, $p$, is equal to $0.3$. The unit cost parameter, $\gamma$ is $0.2$. For each number of MTDs, we apply proposed stochastic coalition formation algorithm in Table~\ref{SIDG} to find out the cooperative groups.
\subsection{Effect of the number of MTDs}
Figs.~\ref{NM:Fail},~\ref{NM:Energy}, and~\ref{PoA:NM} show the effect of the number of MTDs on the stochastic cooperative random access model. In these figures, the number of MTDs is varied from 200 to 1000.

{Fig.~\ref{NM:Fail} shows the network's fail ratio which is defined as the ratio between the number of requests that collided and the total number of requests that all MTDs sent over the RACH during each time slot. From Fig.~\ref{NM:Fail}, we can see that, by increasing the number of MTDs, the fail ratio increases. This is due to the fact that, when the number of MTDs increases, the average arrival rate of the queue in each formed coalition increases. Consequently, sending the access requests over the RACH increases, and, thus, the fail ratio resulting from the stochastic cooperative random access model will increase to the same value of traditional noncooperative random access model. On the average, Fig.~\ref{NM:Fail} shows that the stochastic cooperative random access model reduces the fail ratio of around 14\% and 34\%, for uniform distribution and cluster-based distribution, respectively, compared to the traditional random access protocol, and, on the average, the stochastic cooperative random access model increases the fail ratio of around 17\% and 32\%, for uniform distribution and cluster-based distributions, respectively, compared to the optimal solution. However, the complexity of the optimal solution is significantly larger than the complexity of our proposed distributed algorithm. In particular, the complexity of the optimal solution grows exponentially with the number of MTDs while the required iterations for convergence of the proposed algorithm is restricted to the number of merge-and-split moves per coalition. Thus, the proposed solution offers a better balance between complexity and performance.}
\begin{figure}[t]
\centering
\includegraphics[width=8.0cm,height=2.1in]{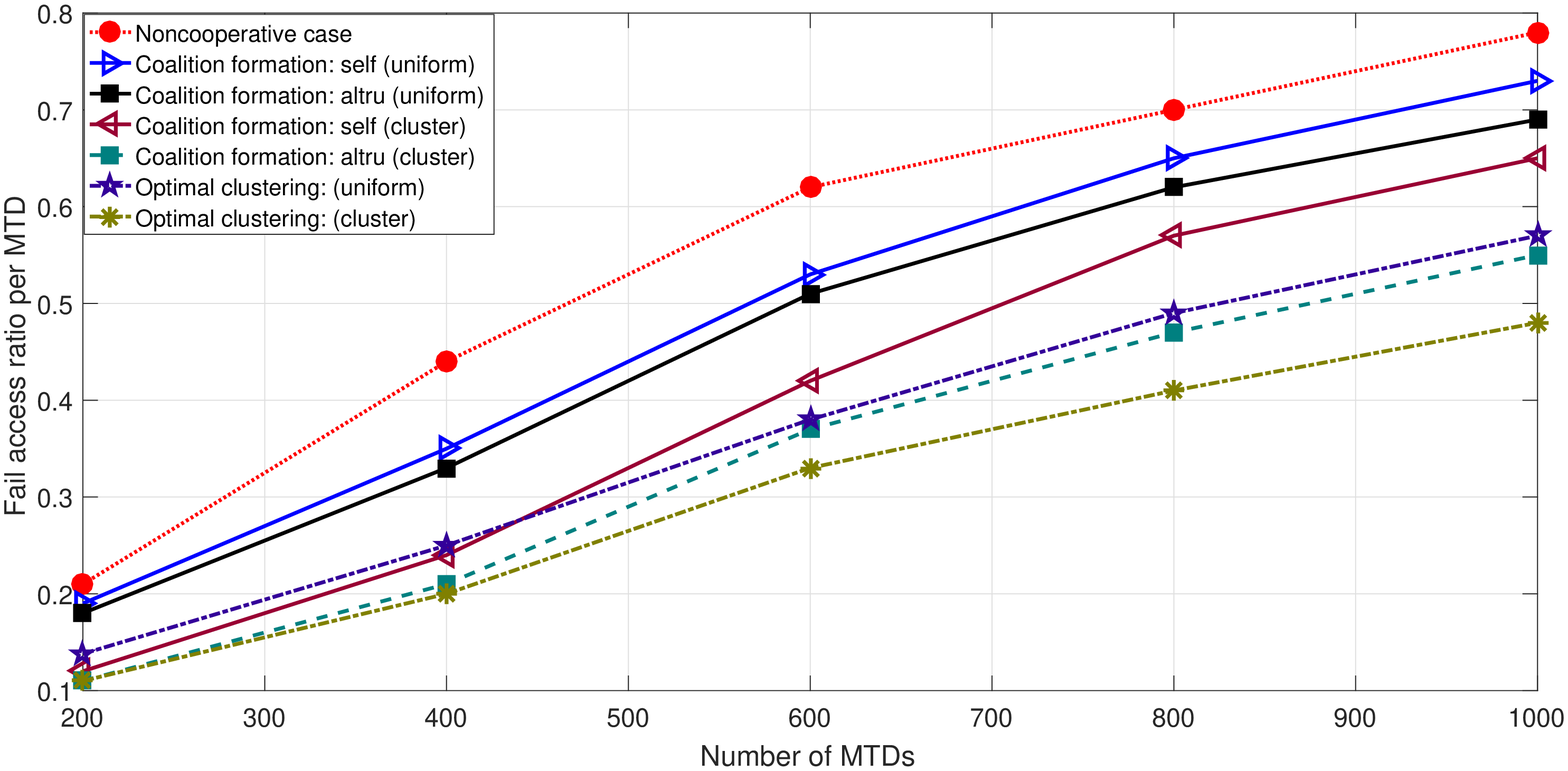}
\caption{Effect of number of MTDs on fail ratio.}
\label{NM:Fail}
\end{figure}

{In Fig.~\ref{NM:Energy}, we show the average per-MTD energy consumption for transmitting over the RACH. From Fig.~\ref{NM:Energy}, we can see that the average per-MTD energy consumption resulting from the proposed stochastic cooperative random access model is less than the one resulting from the traditional noncooperative random access model. The reason is that, in each coalition, all of the MTDs share the energy consumed for acting as head. Fig.~\ref{NM:Energy} also shows that selfish and altruistic coalition formation achieve an almost equal fail ratio and energy consumption. When the MTDs are distributed according to cluster-based distribution, the performance of the proposed coalition formation algorithm improves for both the fail ratio and energy consumption. Fig.~\ref{NM:Energy} also shows that, on the average, the proposed approach reduces the energy consumption of around 16\% and 31\% for uniform distribution and cluster-based distributions, respectively, compared to a traditional random access protocol. Moreover, on the average, the stochastic cooperative random access model consumes around 14\% and 27\% more energy compared to the optimal solution, for the uniform and cluster-based distribution, respectively.}
\begin{figure}[t]
\centering
\includegraphics[width=8.0cm,height=2.1in]{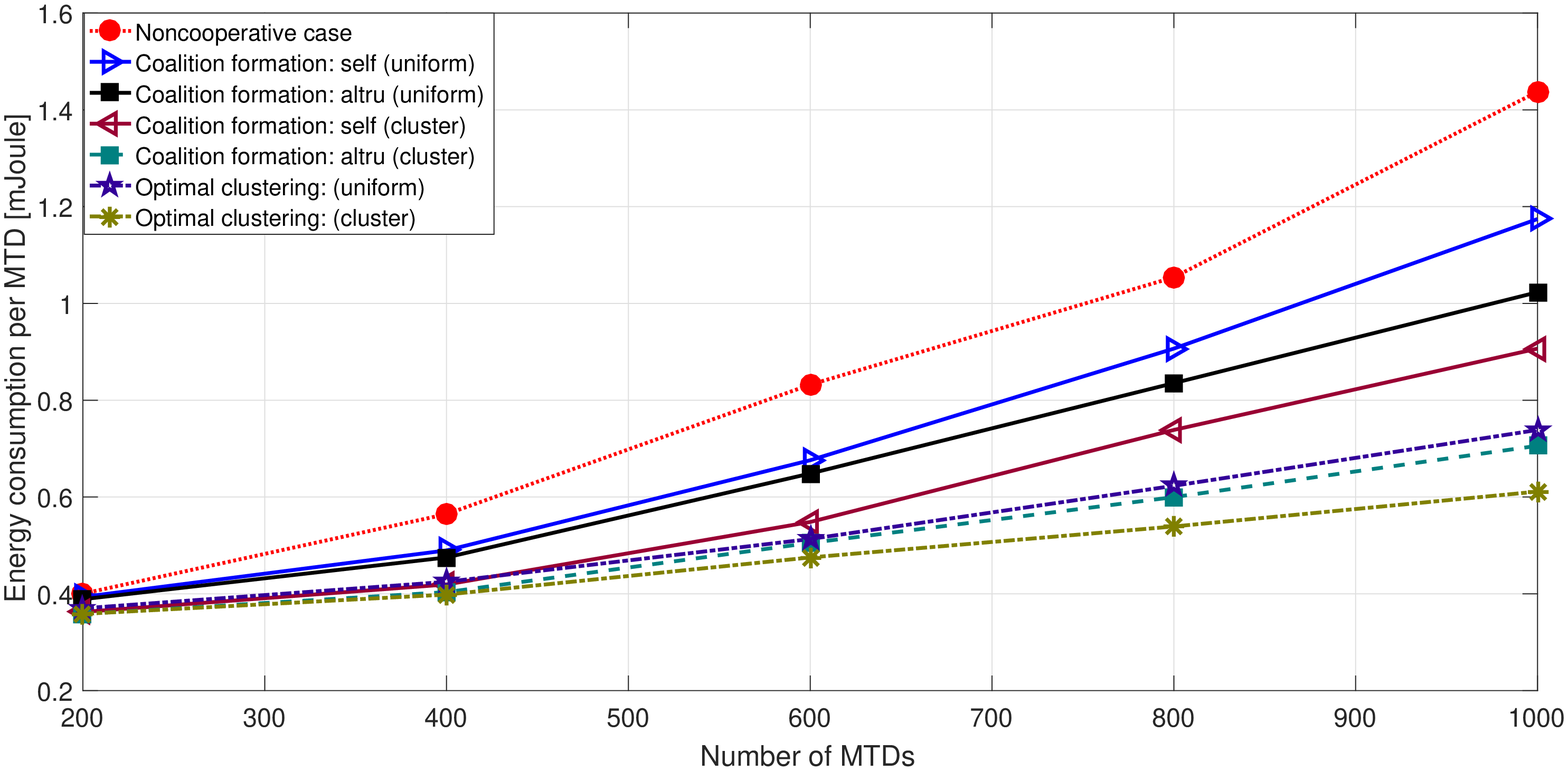}
\caption{Effect of number of MTDs on the energy consumption per MTD.}
\label{NM:Energy}
\end{figure}

{Fig.~\ref{PoA:NM} shows the price of anarchy which is defined as the ratio between the utility achieved by the MTDs at the convergence of the proposed stochastic cooperative random access model and the utility achieved by the MTDs under the optimal solution. Note that we are interested in price of anarchy values that are close to 1 (or $100\%$) in which case the formed cooperative groups at the convergence of the proposed stochastic coalitional game provide a good approximation of the optimal solution.} In Fig.~\ref{PoA:NM}, we can see that, by increasing the number of MTDs, the price of anarchy decreases due to the fact that a a network having more MTDs will provide more opportunities for cooperation. Thus, the number of stable cooperative groups resulting from the proposed stochastic cooperative random access model is greater than the number of clusters under the optimal solution in (\ref{altruprob}). The highest price of anarchy is observed when the MTDs are distributed according to a cluster-based distribution under altruistic cooperation. On the average, Fig.~\ref{PoA:NM} shows that the price of anarchy is 80\% and 93\% (76\% and 83\%) for uniform distribution and cluster-based distribution under altruistic (selfish) cooperation, respectively. Thus, on the average, the performance gap between the proposed distributed algorithm and the optimal solution are 20.5\% and 13.5\% for the uniform and cluster-based distributions, respectively. However, the complexity and signaling overhead of our proposed distributed algorithm are much smaller than the signaling overhead and complexity of the optimal solution.
\begin{figure}[t]
\centering
\includegraphics[width=8.0cm,height=2.1in]{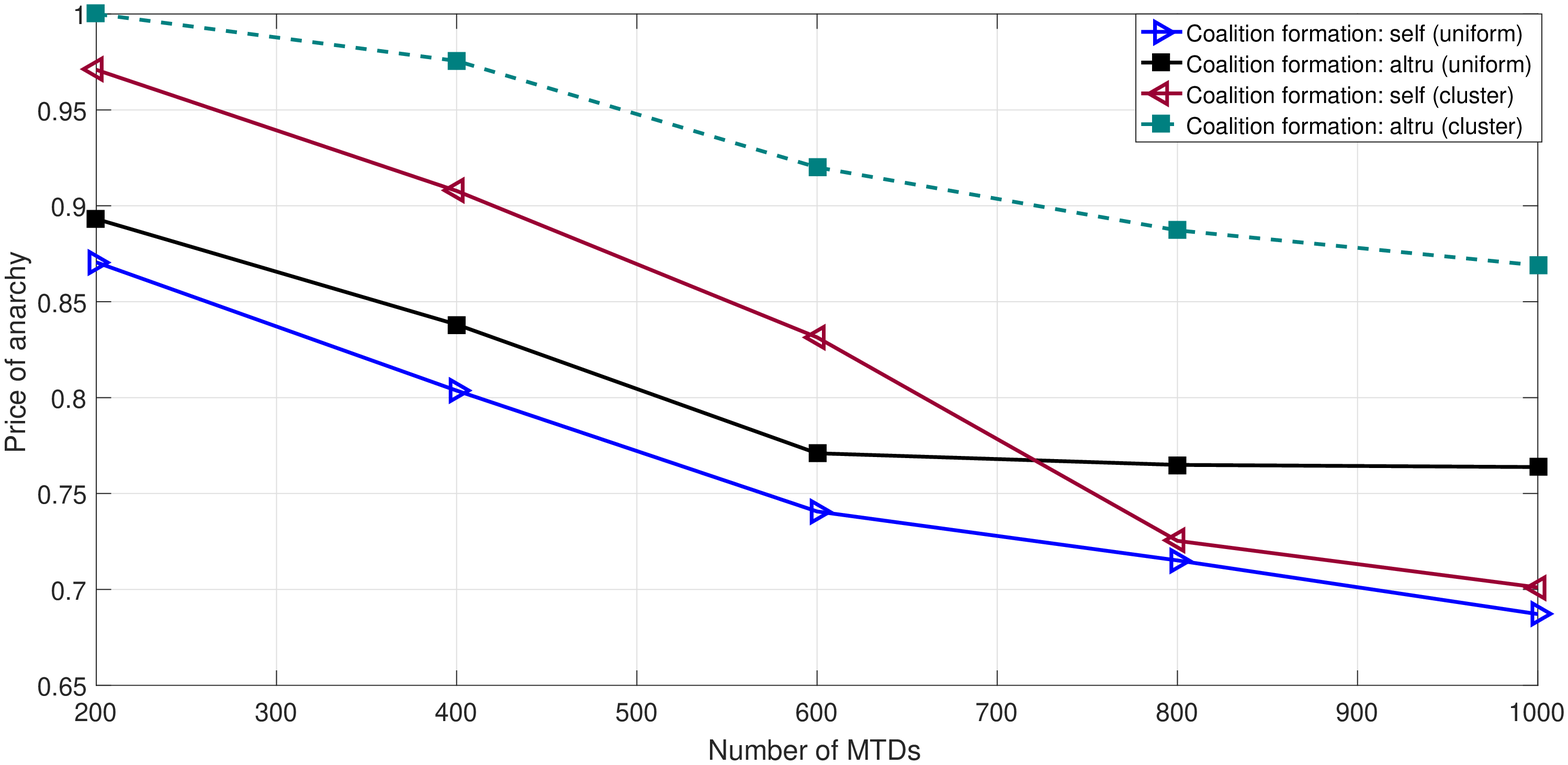}
\caption{Price of anarchy as function of the number of MTDs.}
\label{PoA:NM}
\end{figure}


Fig.~\ref{NM:N_iteration} shows how the number of iterations needed for convergence of the proposed algorithm in Table I, for two different values of the discount factor $\delta$ when the MTDs are altruistic. From this figure, we can see that, when the number of MTDs is 600 for uniform (cluster-based) distribution, the number of iterations will be 150 and 115 (195 and 165) for discount factor values of $0.5$ and $0.8$, respectively. By increasing the number of MTDs, the number of iterations increases.  When the number of MTDs increases, the MTDs will form more coalitions with smaller size. This means that MTDs in a large coalition would prefer to split into smaller ones. In the cluster-based distribution, the size of coalitions is larger than in the uniform distribution. Thus, the number of iterations needed to form stable coalitions is smaller for the cluster-based distribution. When the discount factor increases, the MTDs become more motivated to stay in their coalition for future payoff, thus the number of iterations decreases. Such a number of iterations is quite reasonable for a dense network having thousands of MTDs, as it implies roughly about 3 to 4 iterations per device. As explained next, such a convergence time, which is needed only for the initial coalition formation process, is practical. In essence, an initial delay is required for achieving the stability of all coalitions. In addition, the convergence results shown correspond to a network which starts with an initial state in which each MTD is a singleton, noncooperative player. Such an initial state is, in fact, the worst-case, in terms of convergence time. Once these singleton MTDs form their initial set of coalitions, if there is a need to re-run the coalition formation process due to stochastic changes in environment, the algorithm will be performed starting from the last convergence state not the initial state. Naturally, such a process will require fewer iterations than in the initial state, as the MTDs would have already self-organized into an initial set of coalitions. Moreover, in practice, the changes in the location or incoming traffic of MTDs happen only within a limited area of the network. Thus, there is no need to perform the proposed distributed algorithm for all of the MTDs.
\begin{figure}[t]
\centering
\includegraphics[width=8.0cm,height=2.1in]{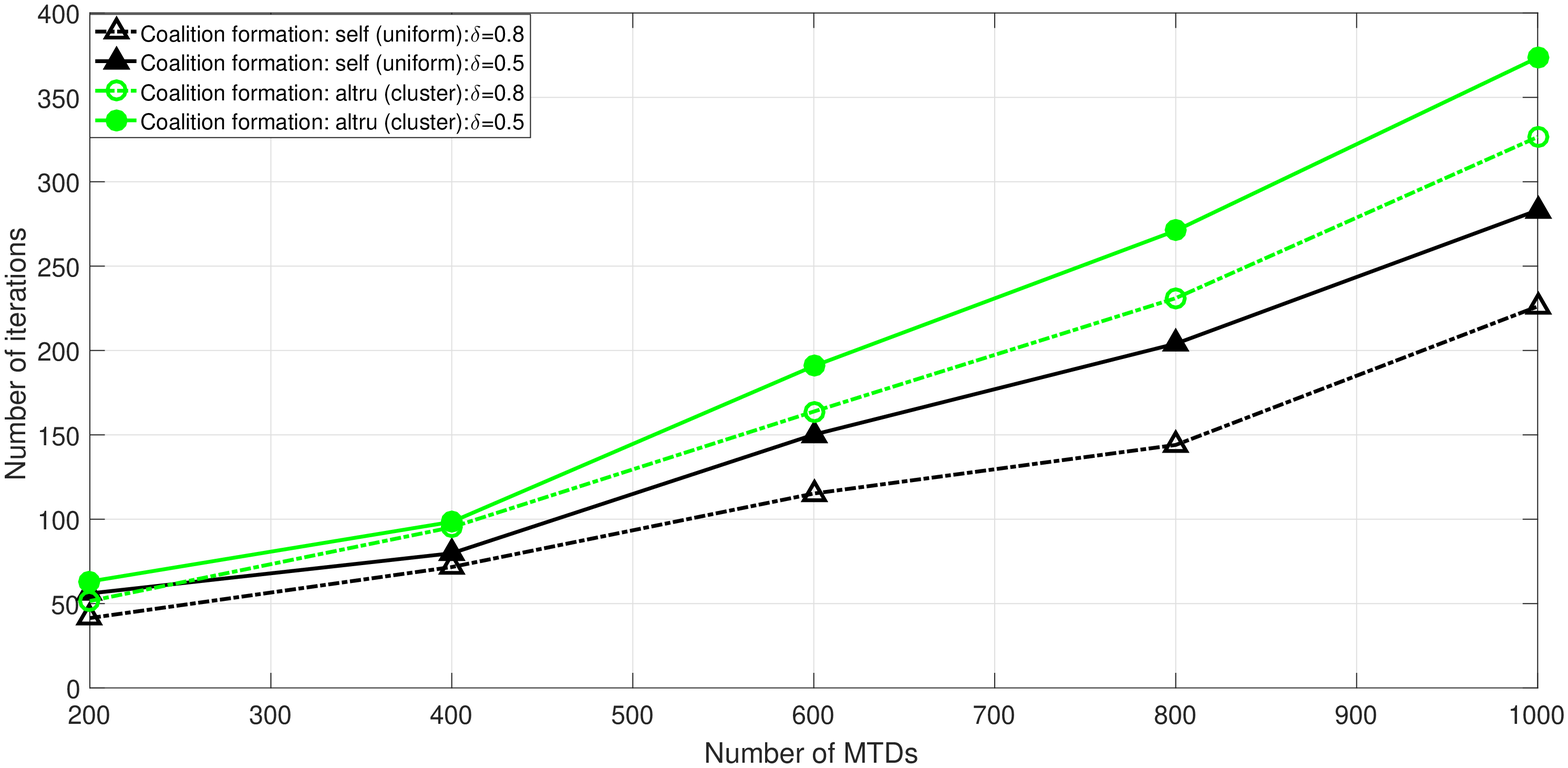}
\caption{Effect of number of MTDs on the number of iterations.}
\label{NM:N_iteration}
\end{figure}

Fig.~\ref{NM:N_move} shows how the average number of moves per coalition changes as a function of the network size and the discount factor, in the case of selfish cooperation. For example, when the number of MTs is 600 for the uniform (cluster-based) distribution, the average number of moves per coalition will be around 1 and 2 (3 and 6) for $\delta = 0.5$ and $\delta = 0.8$, respectively. As the number of MTDs increases, the average number of moves per coalition decreases. Moreover, the average number of moves per coalition becomes smaller for higher values of $\delta$. This decrease in the average number of moves per coalition can be explained as follow. When the number of MTDs increases, the fail ratio (energy consumption) of MTDs achieved by our proposed coalition formation under stochastic changes decreases (increases) (See Figs.~\ref{NM:Fail} and~\ref{NM:Energy}), and the MTDs prefer to form smaller-sized coalitions. This naturally justifies the smaller number of moves per coalition. Moreover, as $\delta$ increases, the MTDs become more farsighted and, as a result, they will be more inclined stay in their coalitions instead of moving to form other coalitions.
\begin{figure}[t]
\centering
\includegraphics[width=8.0cm,height=2.1in]{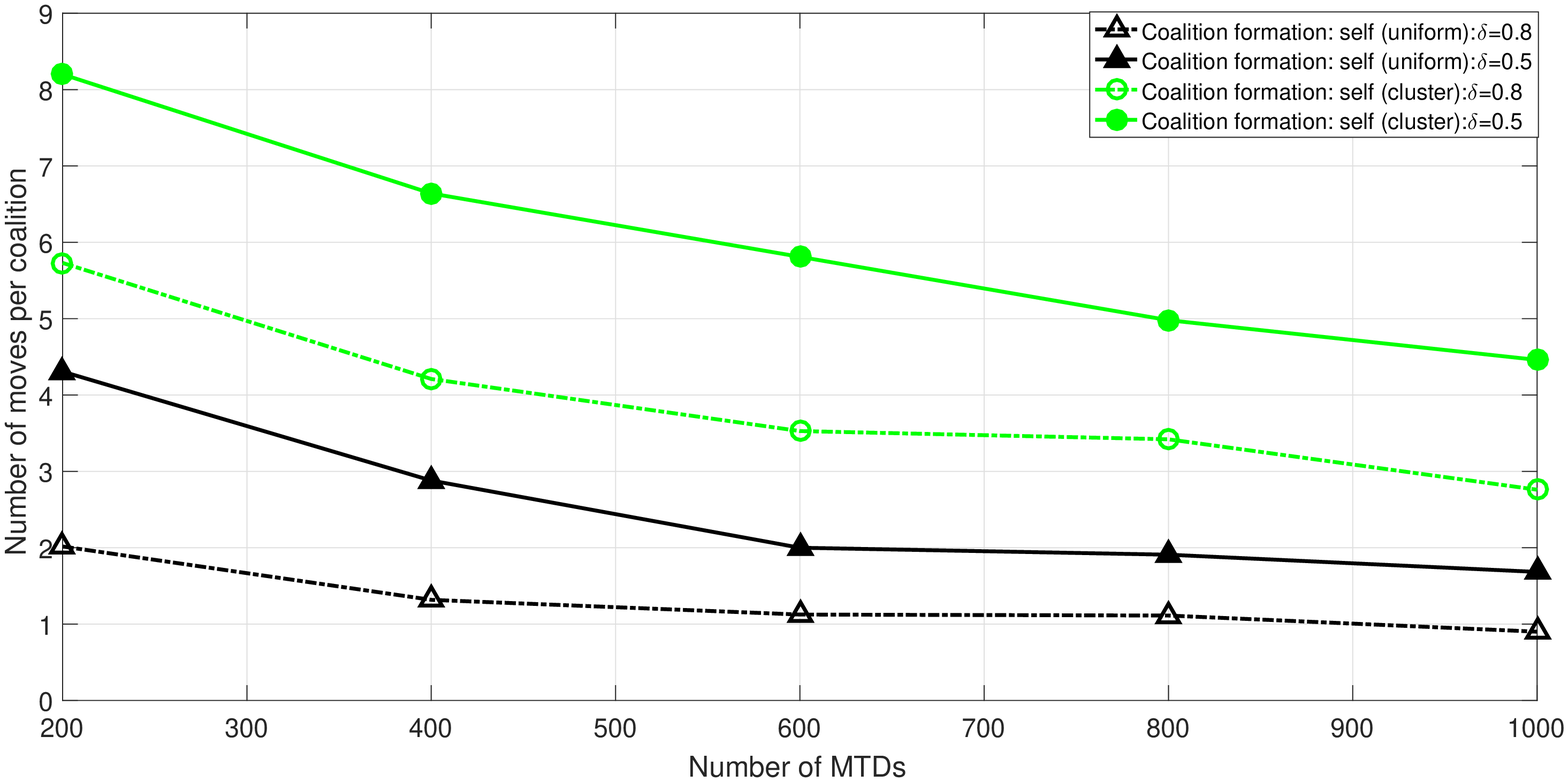}
\caption{Effect of number of MTDs on the number of moves per coalition.}
\label{NM:N_move}
\end{figure}
\subsection{Effect of preferences}
As shown in Figs.~\ref{Fail_Pref} and ~\ref{Energy_Pref}, we  consider three scenarios: Scenario I in which MTDs are more sensitive to the queue length than the energy consumption ($\alpha=0.9,\beta=0.1$), and Scenario II in which MTDs are equally sensitive to the queue length and as energy consumption ($\alpha=0.5,\beta=0.5$), and Scenario III where MTDs favor the queue length less than energy consumption ($\alpha=0.1,\beta=0.9$). To show the effect of the signaling overhead, we analyze the results of the three scenarios for two different values of $\gamma$, i.e., $0.05$ and $0.2$. The number of MTDs is $500$ and the probability of sending requests to the BS is $0.3$.

According to Fig.~\ref{Fail_Pref} and Fig.~\ref{Energy_Pref}, in Scenario I (III) the average fail ratio and energy consumption increase (decrease) under the stochastic cooperative random access model when compared with Scenario II. In addition, when the cost of the signaling overhead decreases, the average fail ratio and energy consumption decrease in all of the scenarios. For example, when $\gamma=0.2$, under the selfish (altruistic) cooperation scheme for uniform (cluster-base) distribution, the average fail ratios are $0.62$, $0.44$, and $0.38$ (0.34, 0.28, and 0.21), and the energy consumption per MTD will be $0.82$, $0.53$, and $5.1$ (0.47, 0.44, and 0.41) mJoules for Scenarios I, II, and III, respectively. Here, we observe that the coalition size decreases (increases) in Scenario I (III) and the number of formed coalitions increases (decreases) in Scenario I (III). An increase in the number of formed coalitions leads to additional load over the RACH and a larger coalition size leads to less energy consumption for the members of the coalition. Consequently, the fail ratio and energy consumption increase for Scenario I, but the fail ratio and energy consumption decrease for Scenario III. When the unit cost parameter of the signaling overhead becomes smaller for the MTDs, the coalition size increases. Thus, the fail ratio and energy consumption decreases when $\gamma=0.05$.\vspace{-0.5cm}
\begin{figure}[t]
\centering
\includegraphics[width=8.0cm,height=2.1in]{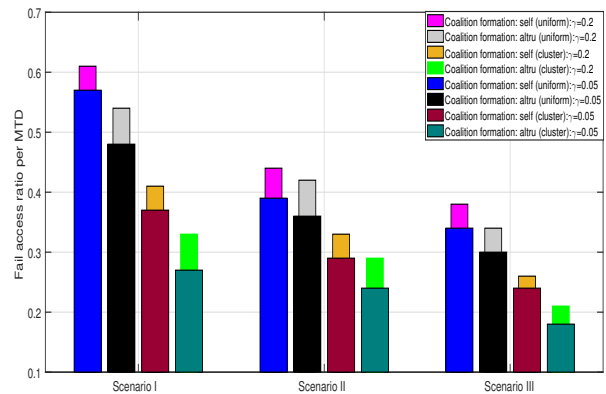}
\caption{Effect of preference of MTDs on fail ratio when $M=500$.}
\label{Fail_Pref}
\end{figure}
\begin{figure}[t]
\centering
\includegraphics[width=8.0cm,height=2.1in]{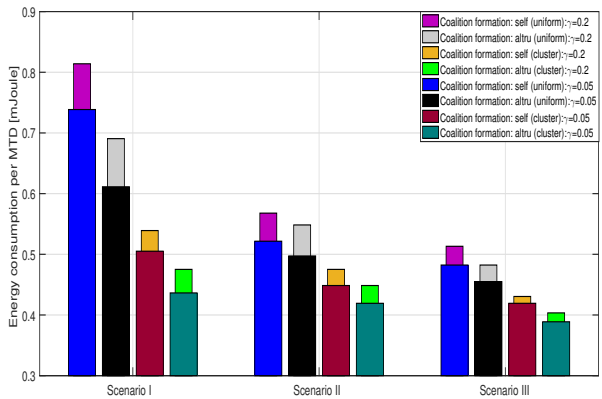}
\caption{Effect of preference of MTDs on energy consumption when $M=500$.}
\label{Energy_Pref}
\end{figure}

\section{Conclusion}
\label{Sec:Conclusion}
In this paper, we have proposed a novel queue length and energy-aware cooperation scheme for optimizing RACH access in M2M cellular networks. In the proposed model, the MTDs autonomously engage in a coalition formation process to coordinate their RACH access by forming cooperative groups. We have modeled the cooperation of MTDs during each time slot as a coalitional game in stochastic characteristic function form. In this game, the MTDs seek to optimize a payoff function that includes two components: a deterministic component related to the certain payoff achieved at the current time slot and a random term that corresponds to the correspond stochastic payoff achieved in future time slots. To model the cooperation of MTDs during different time slots, we have used an $M$-person stochastic coalition game. To solve this game, we have introduced a novel coalition formation algorithm that is shown to reach a stable partition. We have considered not only the energy consumption but also queue length of access requests as the payoff of the MTDs. We have shown that, despite the stochastic changes in the payoff of the MTDs, the MTDs can form stable coalitions, if they are sufficiently farsighted. Simulation results have also shown that on the average, the proposed stochastic coalition formation algorithm can significantly reduce the fail ratio and energy consumption of the system. {For future work, we can analyze the system in presence of multiple base stations and inter-cell interference.}

\bibliographystyle{IEEEtran}
\bibliography{references}

\appendix
\section{}
\subsection{Proof of Theorem~\ref{Theor_band_delta}}
In this appendix, we are going to prove the minimum of $\delta$. Considering the stability conditions in (\ref{stable_cond}) and $\mathcal{C}_k$ in $\mathfrak{C}_m$:
\begin{equation*}
v_m(\mathcal{S}_i^*,Q_{\mathcal{S}_i,t})\geq v_m(\mathcal{C}_k,Q_{\mathcal{C}_k,t}),
\end{equation*}
following (\ref{selfish Value_fun}), we can write:
\begin{align*}
&u_{d,m}({\mathcal{S}_i^*},Q_{\mathcal{S}_i^*,t})+u_{r,m}({\mathcal{S}_i^*},Q_{\mathcal{S}_i^*,t})\geq u_{d,m}({\mathcal{C}_k},Q_{\mathcal{C}_k,t})\nonumber\\
& +u_{r,m}({\mathcal{C}_k},Q_{\mathcal{C}_k,t}),
\end{align*}
according to (\ref{Value_fun}), we can write:
\begin{align*}
& u_{d,m}({\mathcal{S}_i^*},Q_{\mathcal{S}_i^*,t})+
\sum_{n=t+1}^{\infty}\delta^{n-t} \bar{u}_d(\mathcal{S}_i^*,Q_{\mathcal{S}_i*,t\rightarrow n})\geq \nonumber \\
& u_{d,m}({\mathcal{C}_k},Q_{\mathcal{C}_k,t})+
\sum_{n=t+1}^{\infty}\delta^{n-t} \bar{u}_d(\mathcal{C}_k,Q_{\mathcal{C}_k,t\rightarrow n}),
\end{align*}
considering (\ref{U_trans}), we can write:
\begin{align*}
& u_{d,m}({\mathcal{S}_i^*},Q_{\mathcal{S}_i^*,t})+
\sum_{n=t+1}^{\infty}\delta^{n-t} \Delta_{t}^{n} u_d({\mathcal{S}_i^*},Q_{\mathcal{S}_i^*,n})\geq \nonumber \\
& u_{d,m}({\mathcal{C}_k},Q_{\mathcal{C}_k,t})+
\sum_{n=t+1}^{\infty}\delta^{n-t} \Delta_{t}^{n} u_d({\mathcal{C}_k},Q_{\mathcal{C}_k,n}),
\end{align*}
where $\Delta_{t}^{n}$ can be calculated by (\ref{Cal_delta}). By some simplification, we can write:
\begin{align*}
& \sum_{n=t+1}^{\infty}\delta^{n-t}\big(\Delta_{t}^{n} u_d({\mathcal{S}_i^*},Q_{\mathcal{S}_i^*,n})-\Delta_{t}^{n} u_d({\mathcal{C}_k},Q_{\mathcal{C}_k,n})\big)
\geq \nonumber \\
&
u_{d,m}({\mathcal{C}_k},Q_{\mathcal{C}_k,t})-
u_{d,m}({\mathcal{S}_i^*},Q_{\mathcal{S}_i^*,t}).
\end{align*}

Consequently, we can find out the minimum of $\delta$ by considering all of these conditions simultaneously:
\begin{align*}
&\sum_{n=t+1}^{\infty}\delta^{n-t}\big(\Delta_{t}^{n} u_d({\mathcal{S}_i^*},Q_{\mathcal{S}_i^*,n})-\Delta_{t}^{n} u_d({\mathcal{C}_k},Q_{\mathcal{C}_k,n})\big)
\geq\nonumber \\
&
u_{d,m}({\mathcal{C}_k},Q_{\mathcal{C}_k,t})-
u_{d,m}({\mathcal{S}_i^*},Q_{\mathcal{S}_i^*,t})\nonumber \\
&
\forall \mathcal{S}_i^* \in \Pi^*, k=1,2,...,|\mathcal{S}_i^*|-1,\text{ and,}
\nonumber \\
&
\sum_{n=t+1}^{\infty}\delta^{n-t}\big(\Delta_{t}^{n} u_d({\mathcal{S}_i^*},Q_{\mathcal{S}_i^*,n})-\Delta_{t}^{n} u_d(\mathcal{S}_i^*\cup\ \mathcal{S}_j^*,Q_{\mathcal{S}_i^*\cup\ \mathcal{S}_j^*,n})\big)
\geq \nonumber \\
&
u_{d,m}(\mathcal{S}_i^*\cup\ \mathcal{S}_j^*,Q_{\mathcal{S}_i^*\cup\ \mathcal{S}_j^*,t})-
u_{d,m}({\mathcal{S}_i^*},Q_{\mathcal{S}_i^*,t})
\nonumber \\
&
\forall \mathcal{S}_j^* \in \Pi^*, \mathcal{S}_j^*\neq\mathcal{S}_i^*,
\end{align*}

Due to considering all of these conditions together, the minimum of $\delta$ is given.
\end{document}